\documentclass[11pt,a4paper]{article}

\usepackage[margin=1in]{geometry}
\usepackage[T1]{fontenc}
\usepackage{lmodern}
\usepackage{microtype}
\usepackage{setspace}
\usepackage{amsmath,amssymb,amsthm,mathtools,bm}
\usepackage{graphicx}
\usepackage{xcolor}
\usepackage{booktabs}
\usepackage{threeparttable}
\usepackage[authoryear,round,longnamesfirst]{natbib}
\usepackage{hyperref}
\usepackage[nameinlink,noabbrev]{cleveref}
\usepackage{fancyhdr}

\graphicspath{{figures/}}

\newcommand{\rev}[1]{#1}

\newcommand{\note}[1]{}

\newcommand{\papertitle}{Random Cap: Optimal Informationally Robust Delegation}

\newif\ifblindreview
\blindreviewfalse

\ifblindreview
  \author{Anonymous}
  \hypersetup{pdftitle={\papertitle},pdfauthor={Anonymous}}
\else
  \author{Jia Zhiyuan\\[0.35em]
    \normalsize Department of Economics, National University of Singapore\\
    \normalsize \href{mailto:e1374371@u.nus.edu}{e1374371@u.nus.edu}}
  \hypersetup{pdftitle={\papertitle},pdfauthor={Jia Zhiyuan}}
\fi

\title{\papertitle\thanks{I am deeply indebted to my supervisor, Yi-Chun Chen, for continuous guidance throughout this project. I also thank Tan Gan, Zhengqing Gui, Zhenyu Hu, Nenad Kos, Fei Li, Jiangtao Li, Yingkai Li and Xiangqian Yang for helpful discussions.  All errors are my own.}}
\date{August 24, 2026}

\hypersetup{
  colorlinks=true,
  linkcolor=blue!45!black,
  citecolor=blue!45!black,
  urlcolor=blue!55!black
}

\ifblindreview
\else
\fi
\newif\ifshowlinenumbers
\showlinenumbersfalse
\ifshowlinenumbers
  \usepackage[mathlines,switch]{lineno}
  \modulolinenumbers[5]
\fi

\newtheorem{theorem}{Theorem}
\newtheorem{proposition}{Proposition}
\newtheorem{lemma}{Lemma}
\newtheorem{corollary}{Corollary}
\theoremstyle{definition}
\newtheorem{definition}{Definition}
\newtheorem{assumption}{Assumption}

\theoremstyle{remark}
\newtheorem{remark}{Remark}

\crefname{theorem}{Theorem}{Theorems}
\Crefname{theorem}{Theorem}{Theorems}
\crefname{assumption}{Assumption}{Assumptions}
\Crefname{assumption}{Assumption}{Assumptions}
\crefname{proposition}{Proposition}{Propositions}
\Crefname{proposition}{Proposition}{Propositions}
\crefname{lemma}{Lemma}{Lemmas}
\Crefname{lemma}{Lemma}{Lemmas}
\crefname{corollary}{Corollary}{Corollaries}
\Crefname{corollary}{Corollary}{Corollaries}
\crefname{definition}{Definition}{Definitions}
\Crefname{definition}{Definition}{Definitions}

\numberwithin{equation}{section}

\DeclareMathOperator{\Var}{Var}

\DeclareMathOperator*{\argmax}{arg\,max}
\DeclareMathOperator*{\argmin}{arg\,min}

\begin{document}

\maketitle

\begin{abstract}
Are simple delegation rules optimal under ambiguity? We study delegation when
the principal knows the mean, but not the distribution, of the agent's private
information. In a parsimonious quadratic constant-bias environment, the robustly optimal
randomized mechanism is a random cap: the principal draws and reveals an upper
bound, below which the agent chooses freely. Randomization strictly outperforms
every deterministic cap by hedging against cap-specific worst-case
distributions. We characterize random caps through a nondecreasing and concave
expected-action rule and construct the solution using a saddle-point approach.
The worst-case distribution features an exponential survival function over its
continuous region and an atom at the upper endpoint. Under regularity
conditions, the result extends to convex-order ambiguity. Moreover, when the
mean is below the agent's bias, an optimum can be implemented by supplementing the random cap with an incentive-neutral outcome lottery, while pure random caps are strictly suboptimal.
\end{abstract}



\vspace{1.25em}

\ifshowlinenumbers
  \linenumbers
\fi

\section{Introduction}\label{sec:introduction}

Information and control rarely reside in the same hands. In many
organizations, the party best informed about the appropriate action is also
biased in how that action is chosen \citep{holmstrom1978incentives}. Citizens
and legislatures cannot condition fiscal policy on every contemporaneous
spending need, while elected officials who observe those needs may be tempted
to overspend. Trade agreements cannot prescribe the appropriate tariff
response to every possible market condition, while national governments have
protectionist incentives to set tariffs too high. These environments share a
basic commitment--flexibility tradeoff: centralizing the decision preserves
control but wastes private information, whereas granting full discretion uses
that information but exposes the principal to the agent's bias.

Delegation offers a middle ground between centralized decision making and
unrestricted discretion. The principal commits ex ante to a set of permissible
actions and allows the privately informed agent to choose from that set. A
canonical example is a cap: when the agent is biased toward higher actions, she
retains discretion only up to a fixed upper bound. Depending on the
application, such a cap may take the form of a spending limit
\citep{amador2006commitment}, a price cap \citep{alonso2008optimal}, a deficit
ceiling \citep{halac2014fiscal}, or a tariff binding
\citep{amador2013theory}. Beyond illustrating the range of applications, these
studies derive conditions under which caps, or more generally interval
delegation rules, are optimal. Such conditions typically impose joint
restrictions on the parties' preferences and on the distribution of the
agent's private information. 

For our purposes, the key implication, holding preferences fixed, is that the
optimality of a deterministic cap is generally not distribution-free. This
raises the question of why cap-like rules should remain attractive across
heterogeneous information environments. Moreover, existing work gives
randomization a limited role, albeit in two distinct ways. Much of the
delegation literature restricts attention from the outset to deterministic
delegation sets. \citet{kovavc2009stochastic}, by contrast, allow for stochastic
mechanisms but show that, under quadratic payoffs and certain conditions, an
optimal mechanism can be chosen to be deterministic.

To address this problem, we consider a quadratic constant-bias
environment.\footnote{Section~\ref{sec:environment} provides a simple example
showing that, even in this stylized environment, the optimal delegation rule is
sensitive to the state distribution and need not be a cap.}
The agent's ideal action is \(\theta\), whereas the principal's ideal action is
\(\theta-b\), where \(b>0\) measures the agent's upward bias. We introduce
robustness by allowing nature to choose any distribution \(G\) over
\(\theta\in[0,1]\) with known mean \(\mu>b\). The principal commits to an
incentive-compatible randomized direct mechanism
\(\pi\in\Pi^{\mathrm{IC}}\) before nature selects \(G\), and evaluates the
mechanism by its worst-case expected payoff:
\[
\sup_{\pi\in\Pi^{\mathrm{IC}}}
\inf_{G\in\mathcal G_\mu}
\mathbb E_{G,\pi}\!\left[v(a,\theta)\right],
\qquad
\mathcal G_\mu
:=
\left\{
G\in\Delta([0,1]):
\mathbb E_G[\theta]=\mu
\right\}.
\]

\rev{We refer to this criterion as informational robustness. In the spirit of
Wilson's \citeyearpar{wilson1987trading} critique, it evaluates the worst-case
guarantee of mechanisms that are detail-free and therefore portable across
informational environments; see \citet{brooksdu2025simplicity} for a survey.}
One interpretation is that the principal knows an aggregate statistic but not
the full information environment. This is natural when a common regulatory
rule or organizational policy must be applied across jurisdictions, divisions,
or periods with different state distributions, while frequent policy
adjustment is costly.

Our main result is that the robust optimum over all randomized
incentive-compatible mechanisms can be implemented by a \emph{random cap}. The
principal commits ex ante to a distribution over caps. After nature selects
the state distribution, a cap \(Z\) is drawn independently of the state and
revealed to the agent, and the principal delegates the interval \([0,Z]\).
The agent then optimally chooses
\[
a=\min\{\theta,Z\}.
\]
Thus, every realization of the mechanism is an ordinary deterministic cap;
the only additional instrument is ex ante randomization over the cap level.

Literal lotteries over legal caps are uncommon, but the relevant
implementation technology is familiar. Commitment products impose hard
restrictions on access to funds \citep{ashraf2006tying}, while randomized
credit-line extensions have been implemented in large-scale bank field
experiments \citep{aydin2022consumption}. More generally, a random cap can be
implemented by combining familiar request-and-truncation procedures with a
publicly committed lottery over approval ceilings. In applications, this may
correspond to randomized spending allowances, tiered approval of price
increases, supplementary borrowing authorizations, or safeguard and
tariff-binding regimes.

Technically, randomization and robustness reinforce each other. Randomization
protects the principal's robust guarantee by limiting nature's ability to
target a particular cap. This hedging is substantive: the optimal random cap
strictly outperforms every deterministic cap. Robustness, in turn, protects the
simplicity of the mechanism. Although the principal is free to use any
randomized incentive-compatible mechanism, finely tailored state-contingent
action lotteries provide no additional robust value. An optimum can instead be
implemented by randomizing only over ordinary caps.

\rev{The comparative statics reveal a sharp economic separation. A higher
mean expands the agent's full-discretion, or safe-harbor, region and shifts
the cap distribution upward while reducing its variance. By contrast, a
larger bias shifts the cap distribution downward, increases its variance, and
makes the cap bind more often and more tightly.}

Random caps also admit a simple characterization. We show that their
conditional moments are pinned down by a nondecreasing and concave
expected-action rule, with the cap distribution uniquely recoverable from its
slope. Mathematically, the associated restrictions parallel limited liability
and two-sided monotonicity in security design, with concavity selecting exactly
those schedules that can be represented as mixtures of standard debt.

Our methodology is based on the economics of a saddle point constructed
through a double-indifference argument. First, nature's worst-case
distribution, the shifted exponential-tail distribution \(G^*\), makes the
principal indifferent among all deterministic caps in \([\mu-b,1-b]\).
Second, the cap distribution \(R^*\) makes nature indifferent among a
corresponding class of mean-\(\mu\) distributions. Under \(R^*\), the
principal's interim payoff is affine on the region where \(G^*\) places its
continuous mass, and its affine extension also contacts the payoff at the
remaining atom. Each deterministic cap in the support of \(R^*\) is therefore
optimal against \(G^*\), while randomization across those caps protects the
principal against deviations from \(G^*\).
Randomization is valuable not because it improves the principal's payoff under
the least-favorable distribution, but because it reshapes the state-by-state
payoff so that no other feasible distribution can lower the principal's
expected payoff. Thus, the double-indifference construction converts an
infinite-dimensional max--min problem into two matching equalization problems,
with the monotonicity implied by incentive compatibility controlling the
remaining tail.

The method also extends beyond the benchmark model. We replace the mean-only
ambiguity set with a convex-order ambiguity set
\(G\preceq_{\mathrm{cx}}F\) and establish the optimality of a random cap under
regularity conditions. Convex-order ambiguity has a natural learning
interpretation because distributions of posterior means are mean-preserving
contractions of the prior. In delegation, however, this interpretation leads
to a distinct signal-design problem: posterior uncertainty remains payoff
relevant and adds a variance term to the objective. We solve that problem as a
comparison with the main model, thereby clarifying the distinction between
direct distributional ambiguity and full signal design. In the latter problem,
the information designer can fully pool the agent's information and eliminate
the commitment--flexibility tradeoff.

\paragraph{Related literature.}
Our paper contributes first to the theory of delegation. Classic work studies
how a principal restricts the discretion of a privately informed and biased
agent and provides conditions under which an interval is optimal among
deterministic delegation rules
\citep{holmstrom1978incentives,melumad1991communication,
amador2006commitment,alonso2008optimal,amador2013theory,
halac2020commitment}.
Fewer papers study stochastic delegation. Under their maintained regularity
conditions, \citet{kovavc2009stochastic} characterize stochastic
incentive-compatible mechanisms and show that an optimum can be chosen to be
deterministic in quadratic settings. Subsequent work develops majorization
methods for stochastic mechanisms and delegation
\citep{kleiner2021extreme,kolotilin2025persuasion}. We depart from this
literature by taking a robustness approach.

Second, the paper contributes to robust mechanism design. A large literature
studies mechanisms that perform well under uncertainty about distributions or
information structures. \citet{carroll2015robustness} studies robust moral
hazard, while \citet{bergemann2011robust,roesler2017buyer} study informational
robustness in monopoly problems, with extensions to multiple-buyer settings
\citep{du2018robust,brooks2024structure,zhang2022random,chen2023information} and multidimensional
environments \citep{deb2021multi,che2021robustly}. More closely related to our
formulation, \citet{carrasco2018optimal,he2024interim,chen2024screening} study
ambiguity sets defined by means or other moments. We bring this moment-based
approach to a mechanism-design problem without transfers. The absence of
transfers makes incentive compatibility a substantive restriction on
state-contingent actions, while the mean constraint gives the adversary's
problem a useful affine-duality structure.

Finally, our paper relates to recent work on robust delegation.
\citet{frankel2014aligned} studies uncertainty about the agent's preferences
and randomization over delegation sets in a multitask environment, while
\citet{AlonsoGanHu2026} study preference robustness and establish the
optimality of convex delegation sets within a class of deterministic compact
sets. In both papers, robustness concerns how the agent ranks available
actions. We instead hold preferences fixed and introduce ambiguity about the
distribution of the agent's information.

Most closely related, \citet{HuLi2023} consider mean-constrained
distributional ambiguity in delegation but restrict attention to deterministic
rules. Parallel ongoing work by \citet{HeLiLiYe2026} study
randomized delegation when the prior distribution is unrestricted. To our
knowledge, our paper is the first to establish the optimality of a random cap
among all stochastic incentive-compatible mechanisms, and our saddle-point
argument permits an extension to convex-order ambiguity.

The remainder of the paper proceeds as follows.
Section~\ref{sec:environment} introduces the environment and characterizes
random-cap implementation.
Section~\ref{sec:heuristic} constructs the candidate saddle point through
double indifference.
Section~\ref{sec:main} proves global optimality over all randomized
incentive-compatible mechanisms.
Section~\ref{sec:extensions} extends the analysis to
mean-preserving-contraction ambiguity, distinguishes the direct distributional
problem from adversarial signal design, quantifies the value of randomization,
and studies the boundary regime.
Section~\ref{sec:discussion} develops the financial interpretation and
discusses the scope and limitations of the analysis.

\section{Environment}\label{sec:environment}

We study a one-shot delegation problem with quadratic loss. The state is
\(\theta\in[0,1]\). The agent privately observes the realized state, while the
principal commits ex ante to an incentive-compatible direct mechanism that
governs the agent's action. There are neither transfers nor participation
constraints.

\subsection{Preferences and feasibility}\label{subsec:prefs}

The action space is \(A:=\mathbb R\). The parameter \(b>0\) measures the
agent's upward bias relative to the principal. Their ex post payoffs are
\[
v(a,\theta)=-\bigl(\theta-(a+b)\bigr)^2,
\qquad
u(a,\theta)=-(\theta-a)^2.
\]

\subsection{Mechanisms and incentive compatibility}\label{subsec:mech}

By the revelation principle, it is without loss of generality to consider
direct mechanisms. A direct randomized mechanism assigns to every report
\(\hat\theta\in[0,1]\) a lottery over actions with a finite second moment. We
represent this lottery directly by its conditional cdf
\[
\pi(a\mid\hat\theta)
:=\Pr(\widetilde a\leq a\mid\hat\theta),
\qquad a\in A,
\]
where \(\widetilde a\) denotes the induced random action. Under quadratic
preferences, it is useful to define the induced conditional mean and variance,
\[
m^\pi(\hat\theta)
:=\int_A a\,\pi(da\mid\hat\theta),
\qquad
q^\pi(\hat\theta)
:=\int_A\bigl(a-m^\pi(\hat\theta)\bigr)^2
\,\pi(da\mid\hat\theta).
\]

If the true state is \(\theta\) and the agent reports \(\hat\theta\), the
agent's and principal's expected payoffs are, respectively,
\[
\begin{aligned}
U^\pi(\theta,\hat\theta)
&:=\int_A u(a,\theta)\,\pi(da\mid\hat\theta)
=-\bigl(\theta-m^\pi(\hat\theta)\bigr)^2-q^\pi(\hat\theta),
\\
V^\pi(\theta,\hat\theta)
&:=\int_A v(a,\theta)\,\pi(da\mid\hat\theta)
=-\bigl(\theta-b-m^\pi(\hat\theta)\bigr)^2-q^\pi(\hat\theta).
\end{aligned}
\]

Let \(\Pi^{IC}\) denote the set of direct randomized mechanisms satisfying
\begin{equation}\label{eq:IC}
U^\pi(\theta,\theta)
\geq U^\pi(\theta,\hat\theta)
\qquad
\text{for all }\theta,\hat\theta\in[0,1].
\end{equation}
We write
\(U^\pi(\theta):=U^\pi(\theta,\theta)\) and
\(V^\pi(\theta):=V^\pi(\theta,\theta)\) for the truthful-reporting payoffs.

\subsection{Random caps and their characterization}
\label{subsec:randomcap}

Optimal deterministic delegation rules often take an interval form
(e.g., \citealp{holmstrom1980theory,alonso2008optimal,amador2013theory,kolotilin2025persuasion}).
When the agent is biased toward higher actions, the interval typically has an
upper endpoint and therefore induces a cap, often interpreted as a spending
limit \citep{amador2006commitment,amador2013theory,halac2020commitment}. This
conclusion, however, is not distribution-free: existing results impose shape
restrictions involving both preferences and the distribution of the agent's
private information. The two-point example below shows that, after a change in
the prior, a disconnected delegation set can strictly dominate every cap.

For a deterministic cap \(z\in[0,1]\), the principal delegates the interval
\([0,z]\), and the agent optimally chooses
\(a_z(\theta)=\min\{\theta,z\}\).

\begin{definition}[Random cap mechanism]\label{def:randomcap}
Let \(R\) be a cdf supported on \([0,1]\). A \emph{random cap mechanism} draws
\(Z\sim R\), independently of \(\theta\), reveals \(Z\) to the agent, and
delegates the interval \([0,Z]\). The induced action is
\[
a(\theta,Z)=\min\{\theta,Z\}.
\]
\end{definition}

A random cap is incentive compatible. Indeed, conditional on \(Z=z\), a
report \(\hat\theta\) induces the action \(\min\{\hat\theta,z\}\), while a
truthful report induces the projection of \(\theta\) onto \([0,z]\) and hence
maximizes the agent's payoff. This argument applies for every realized \(z\).

Let \(\Pi_{\mathrm{RC}}\subseteq\Pi^{IC}\) denote the set of random-cap
mechanisms. The corresponding direct mechanism has conditional cdf
\begin{equation}\label{eq:pi-R}
\pi(a\mid\hat\theta)
=\Pr\bigl(\min\{\hat\theta,Z\}\leq a\bigr)
=
\begin{cases}
R(a), & a<\hat\theta,\\[4pt]
1, & a\geq\hat\theta.
\end{cases}
\end{equation}

For a known prior, the Bayesian payoff generated by a random cap is an average
of the payoffs generated by deterministic caps. Thus, within the random-cap
class, randomization cannot strictly improve upon the best deterministic cap.
This observation helps explain why random caps have received little attention
in the Bayesian delegation literature. Our first result instead gives a
mechanism-design characterization of this class.

We say that a direct mechanism \(\pi\) is \emph{implementable by a random cap}
if a random cap reproduces \((m^\pi,q^\pi)\) at every state. Because both
players have quadratic payoffs, this is equivalent to state-by-state payoff
equivalence for both players.

\begin{proposition}[Random-cap representability]\label{prop:imple}
Let \(m,q:[0,1]\to\mathbb R\) be Borel measurable. There exists a random
variable \(Z\) supported on \([0,1]\) such that
\[
m(\theta)=\mathbb E\bigl[\min\{\theta,Z\}\bigr],
\qquad
q(\theta)=\operatorname{Var}\bigl(\min\{\theta,Z\}\bigr)
\qquad
\text{for every }\theta\in[0,1]
\]
if and only if the following conditions hold:
\begin{enumerate}
    \item[(a)] \(m(0)=0\);
    \item[(b)] \(m\) is nondecreasing and concave on \([0,1]\), and
    \(m(\theta)\leq\theta\) for every \(\theta\in[0,1]\);
    \item[(c)] for every \(\theta\in[0,1]\),
    \[
    q(\theta)
    =2\theta m(\theta)-2\int_0^\theta m(t)\,dt-m(\theta)^2.
    \]
\end{enumerate}
Moreover, the distribution of \(Z\) is unique among distributions supported on
\([0,1]\). Its cdf is
\[
R(z)=
\begin{cases}
0, & z<0,\\[3pt]
1-m'_+(z), & 0\leq z<1,\\[3pt]
1, & z\geq1,
\end{cases}
\]
where \(m'_+\) denotes the right derivative of \(m\).
\end{proposition}

Appendix~\ref{app:random-cap-proof} proves the proposition. The standard
envelope argument characterizes incentive-compatible mechanisms in terms of
monotonicity of \(m\), nonnegativity of \(q\), and an envelope identity; see
\citet{kovavc2009stochastic} and Lemma~\ref{lem:IC}.

This proposition also links random-cap mechanisms to the
costly-state-verification literature and to financial contracting. This
connection provides a natural economic interpretation of random caps, which
we develop in detail in Section~\ref{subsec:CSV}.

\subsection{Information and robustness}\label{subsec:robust}

After the principal commits to a mechanism, nature selects a
distribution \(G\) over \(\theta\) subject to a mean constraint. We consider the
ambiguity set
\[
\mathcal G_\mu
:=\left\{
G\in\Delta([0,1])
:\int_0^1\theta\,dG(\theta)=\mu
\right\}.
\]
A more general distributional constraint is considered in
Section~\ref{sec:extensions}.

\begin{assumption}\label{ass:mu-range}
\(0<b<\mu<1\).
\end{assumption}

Assumption~\ref{ass:mu-range} isolates the main interior regime, in which
the bias is smaller than the mean state. In the uniform-prior Bayesian
benchmark, \(\mu=1/2\), and \(b\geq\mu\) leads to full restriction within the
nonnegative cap class. We later relax the assumption in
Section~\ref{subsec:boundary}.

The main robustness problem we will study is
\begin{equation}\label{eq:maxmin}
\sup_{\pi\in\Pi^{IC}}\ \inf_{G\in\mathcal G_\mu}\
\mathbb E_{\theta\sim G}\bigl[V^\pi(\theta)\bigr].
\end{equation}
We also study the \emph{min--max} benchmark
\begin{equation}\label{eq:minmax}
\inf_{G\in\mathcal G_\mu}\ \sup_{\pi\in\Pi^{IC}}\
\mathbb E_{\theta\sim G}\bigl[V^\pi(\theta)\bigr],
\end{equation}
in which nature moves first and the inner supremum is the principal's best
response to the known distribution \(G\).

To illustrate why robustness matters, first consider a Bayesian benchmark with
a uniform prior and bias \(b=0.1\). As shown by
\citet{kovavc2009stochastic}, the optimal delegation rule in this case is the
cap \([0,z_U]\), where \(z_U=1-2b=0.8\). Now replace the uniform prior by the
two-point distribution
\[
F=\frac{1}{2}\delta_{1/2}+\frac{1}{2}\delta_1,
\]
while keeping preferences unchanged, and consider the disconnected delegation
set
\[
D=\{1/2-b,\,1-b\}=\{0.4,0.9\}.
\]
The low type chooses \(0.4\), the high type chooses \(0.9\), and the
principal's ideal action is implemented in both states. By contrast, no cap
can implement both ideal actions. A cap high enough to implement \(0.9\) at
\(\theta=1\) leaves the type \(\theta=1/2\) unconstrained at \(1/2\),
whereas a cap low enough to implement \(0.4\) at \(\theta=1/2\) also
constrains the high type. Indeed, the best cap is \([0,0.9]\): it implements
the principal's ideal action for the high type but generates a loss of
\(b^2\) for the low type, and hence an expected loss of
\(\frac{1}{2}b^2=0.005\).

Moreover, the failure of cap optimality is not an artifact of atomic
distributions; it persists under strictly positive, absolutely continuous
densities. Consider a sequence of atomless, full-support priors
\(\{F_n\}_{n\geq 2}\), each with a strictly positive and absolutely continuous
density on \([0,1]\), such that
\[
F_n
\Rightarrow
\frac{1}{2}\delta_{1/2}+\frac{1}{2}\delta_1
\qquad\text{as } n\to\infty.
\]
Such a sequence can be constructed explicitly.\footnote{For example, let
\[
f_n(\omega)
=
\frac{1}{n}
+
\frac{1-\frac{1}{n}}{2}
\left[
\frac{\omega^{n-1}(1-\omega)^{n-1}}{B(n,n)}
+
n\omega^{n-1}
\right],
\qquad \omega\in[0,1],
\]
where \(B(\cdot,\cdot)\) denotes the beta function. The two terms in brackets
are the densities of a \(\operatorname{Beta}(n,n)\) distribution and a
\(\operatorname{Beta}(n,1)\) distribution, respectively. The former
concentrates at \(1/2\), while the latter concentrates at \(1\), as
\(n\to\infty\). The first term, \(1/n\), is a vanishing uniform component
that ensures \(f_n(\omega)>0\) for every \(\omega\in[0,1]\). Since
\[
\int_0^1 f_n(\omega)\,d\omega
=
\frac{1}{n}
+
\frac{1-\frac{1}{n}}{2}(1+1)
=
1,
\]
\(f_n\) is a density. Moreover, each \(f_n\) is continuously differentiable,
and hence absolutely continuous, on \([0,1]\). It follows that
\[
F_n
\Rightarrow
\frac{1}{2}\delta_{1/2}+\frac{1}{2}\delta_1.
\]}
Under the disconnected set \(D=\{0.4,0.9\}\), types near \(1/2\) choose
\(0.4\), while types near \(1\) choose \(0.9\). The principal's expected
loss under \(D\) therefore converges to zero as \(n\to\infty\). By contrast,
the minimum expected loss attainable by a cap converges to \(0.005\), the
minimum cap loss under the limiting two-point prior. To see this, note that
the loss induced by a cap is bounded and jointly continuous in the state and
the cap level. Weak convergence of the priors therefore implies uniform
convergence of cap losses over the compact set of cap levels
\(z\in[0,1]\). The strict payoff gap consequently persists for all
sufficiently large \(n\): the disconnected set \(D\) strictly outperforms
every cap.

\section{A heuristic construction via double indifference}
\label{sec:heuristic}

This section uses two complementary indifference conditions to construct a
candidate saddle-point pair. The derivation is deliberately heuristic: on the
intervals where derivatives are taken, we temporarily suppose that the
relevant distribution and interim payoff are sufficiently regular. These
auxiliary smoothness assumptions serve only to reveal the construction. The
lemmas below state the resulting properties exactly, while
Section~\ref{sec:main} establishes global optimality over the full class of
incentive-compatible mechanisms and verifies the saddle point.

\paragraph{Heuristic derivation}
Consider first the min--max problem in \eqref{eq:minmax}. Fix a distribution
\(G\) and temporarily restrict the principal to deterministic caps. Under a
cap \(z\in[0,1]\), the agent chooses \(a=\min\{\theta,z\}\). Replacing the
full mechanism class with deterministic caps gives the benchmark
\begin{equation}\label{eq:minmax-det-cap}
\inf_{G\in\mathcal G_\mu}\ \sup_{z\in[0,1]}\ \phi_z(G),
\end{equation}
where
\begin{align}
\phi_z(G)
&:=-\mathbb E_G\!\left[
\bigl(\min\{\theta,z\}-(\theta-b)\bigr)^2
\right]
\label{eq:phi-def}\\
&=-b^2G(z)
-\int_{(z,1]}(z+b-\theta)^2\,dG(\theta).
\nonumber
\end{align}
Suppose, for the purpose of this calculation, that \(G\) has no atom on the
interior region over which \(z\) varies and that \(\phi_z(G)\) is
differentiable there. Then
\begin{align}
\frac{\partial}{\partial z}\phi_z(G)
&=-2\int_{(z,1]}(z+b-\theta)\,dG(\theta)\nonumber\\
&=-2[1-G(z)]
\left(z+b-\mathbb E_G[\theta\mid\theta>z]\right).
\label{eq:FOC}
\end{align}
Thus, to make the principal indifferent over an interval of interior caps, we
seek a distribution satisfying
\begin{equation}\label{eq:heuristic-tail-condition}
z+b=\mathbb E_G[\theta\mid\theta>z]
\quad\Longleftrightarrow\quad
\int_z^1[1-G(x)]\,dx=b[1-G(z)]
\end{equation}
throughout that interval. That is, the current cap equals the principal's
average ideal action among the types constrained by the cap. Formally
differentiating the second equality gives an exponential survival function.

For the complementary equalizer, fix a mechanism \(\pi\) and temporarily
suppose that its interim payoff \(V^\pi\) is continuously differentiable.
Integration by parts gives
\begin{equation}\label{eq:payoff-tail}
\mathbb E_G[V^\pi(\theta)]
=V^\pi(0)+\int_0^1(V^\pi)'(x)[1-G(x)]\,dx,
\end{equation}
whereas the mean constraint is
\begin{equation}\label{eq:mean-tail}
\int_0^1[1-G(x)]\,dx=\mu.
\end{equation}
Hence, if \((V^\pi)'\) is constant on an interval, redistributions of the
survival function within that interval leave expected payoff unchanged as
long as they preserve its contribution to the mean. Equivalently,
\(V^\pi\) should be affine on the region over which nature is to be made
indifferent. Because the candidate distribution suggested above also has an
isolated upper-support point at \(1\), the affine segment must extend to the
point \((1,V^\pi(1))\).

The same observation has a useful dual interpretation. Nature's problem
against a fixed mechanism is a linear moment problem whose dual searches for
an affine lower bound
\[
\ell(\theta)=\alpha+\beta\theta\leq V^\pi(\theta).
\]
Such a bound implies
\[
\mathbb E_G[V^\pi(\theta)]\geq\ell(\mu)
\qquad\text{for every }G\in\mathcal G_\mu.
\]
If \(V^\pi\) coincides with \(\ell\) on a set \(S\), every mean-\(\mu\)
distribution supported on \(S\) attains the same payoff \(\ell(\mu)\).
This dual geometry motivates a mechanism whose interim payoff is affine on
\([\mu-b,1-b]\) and whose affine extension also contacts the payoff at
\(\theta=1\). Section~\ref{sec:main} verifies that the resulting affine
function is indeed a global lower bound.

\paragraph{The candidate shifted exponential-tail distribution}
The first heuristic yields the following exact construction.

\begin{lemma}[Equalizing conditional mean]\label{lem:Gstar-FOC}
Define
\begin{equation}\label{eq:Gstar-x0}
G^*(\theta)=
\begin{cases}
0, & \theta<\mu-b,\\[4pt]
1-e^{-\frac{\theta-\mu+b}{b}},
& \mu-b\leq\theta<1-b,\\[8pt]
1-e^{-\frac{1-\mu}{b}},
& 1-b\leq\theta<1,\\[8pt]
1, & \theta\geq1.
\end{cases}
\end{equation}
Then \(G^*\in\mathcal G_\mu\), and
\[
\mathbb E_{G^*}[\theta\mid\theta>z]=z+b
\qquad
\text{for every }z\in[\mu-b,1-b].
\]
\end{lemma}

Indeed, for every \(z\in[\mu-b,1-b]\), the survival function of \(G^*\)
satisfies
\[
\int_z^1[1-G^*(x)]\,dx=b[1-G^*(z)],
\]
which gives the stated conditional mean. Moreover,
\[
\mathbb E_{G^*}[\theta]
=\int_0^1[1-G^*(x)]\,dx
=\mu,
\]
so \(G^*\) satisfies the mean restriction. Assumption~\ref{ass:mu-range}
ensures that \(0<\mu-b<1-b\), and hence that the construction lies in the
state space.

Under \(G^*\), the principal's payoff from a deterministic cap is
\begin{equation}\label{eq:phi-under-Gstar}
\phi_z(G^*)=
\begin{cases}
-b^2\!\left(1-e^{-\frac{1-\mu}{b}}\right)
-\bigl(z-(\mu-b)\bigr)^2,
& 0\leq z\leq\mu-b,\\[6pt]
-b^2\!\left(1-e^{-\frac{1-\mu}{b}}\right),
& \mu-b\leq z\leq1-b,\\[6pt]
-b^2\!\left(1-e^{-\frac{1-\mu}{b}}\right)
-e^{-\frac{1-\mu}{b}}\bigl(z-(1-b)\bigr)^2,
& 1-b\leq z\leq1.
\end{cases}
\end{equation}
Thus, the principal is indifferent among all deterministic caps in
\([\mu-b,1-b]\). We take
\begin{equation}\label{eq:V-underbar}
V^*
:=-\operatorname{Var}_{G^*}(\theta)
=-b^2\!\left(1-e^{-\frac{1-\mu}{b}}\right)
\end{equation}
as the candidate robust value.

\paragraph{The candidate mechanism}
The integration-by-parts calculation and its affine-dual interpretation
motivate the following exact statement.

\begin{lemma}[Affine equalization]\label{lem:pistar-const}
Define the conditional cdf
\begin{equation}\label{eq:pistar}
\pi^*(a\mid\theta)=
\begin{cases}
0, & a<\min\{\theta,\mu-b\},\\[6pt]
\dfrac{1}{2}e^{\frac{a+b-1}{b}},
& \mu-b\leq a<\min\{\theta,1-b\},\\[10pt]
1, & a\geq\min\{\theta,1-b\}.
\end{cases}
\end{equation}
Then \(V^{\pi^*}\) is affine on \([\mu-b,1-b]\), and the affine extension
of this segment passes through \((1,V^{\pi^*}(1))\).
\end{lemma}

Let \(\ell\) denote this affine extension. The lemma implies that
\(V^{\pi^*}=\ell\) on \([\mu-b,1-b]\cup\{1\}\). Consequently, all
\(G\in\mathcal G_\mu\) supported on this set yield the same expected payoff:
\[
\mathbb E_G[V^{\pi^*}(\theta)]
=\mathbb E_G[\ell(\theta)]
=\ell(\mu).
\]
The formal payoff calculation and the global minorant inequality are deferred
to Section~\ref{sec:main}.

The next corollary uses Proposition~\ref{prop:imple} to express \(\pi^*\) as a
random cap.

\begin{corollary}[Implementation]
\label{cor:pistar-as-RC}
The mechanism \(\pi^*\) in Lemma~\ref{lem:pistar-const} is implemented by
drawing a cap \(Z\) from the cdf
\begin{equation}\label{eq:Rstar}
R^*(z)=
\begin{cases}
0, & z<\mu-b,\\[4pt]
\dfrac{1}{2}e^{\frac{z+b-1}{b}},
& \mu-b\leq z<1-b,\\[10pt]
1, & z\geq1-b,
\end{cases}
\end{equation}
and setting \(a=\min\{\theta,Z\}\).
\end{corollary}

An alternative derivation explains the exponential shape of the cap
distribution. Let \(R\) be an arbitrary cap distribution. At any state
\(\theta\) at which \(R\) is differentiable and has no atom, the principal's
interim payoff is
\[
V^R(\theta)
=-\int_{[0,\theta)}(z+b-\theta)^2\,dR(z)
-b^2[1-R(\theta)].
\]
Differentiating twice gives
\[
(V^R)''(\theta)
=2\bigl[bR'(\theta)-R(\theta)\bigr].
\]
Hence, requiring the interim payoff to be affine on the equalization region
imposes the differential equation \(bR'(\theta)=R(\theta)\), whose solutions
are exponential. The candidate \(R^*\) satisfies this equation on
\((\mu-b,1-b)\).

\paragraph{The double-indifference logic}
Figure~\ref{fig:saddle} plots the cdfs of \(G^*\) and \(R^*\) for
\((\mu,b)=(0.5,0.25)\). Both distributions have lower support point
\(\mu-b\). The cdf \(G^*\) rises continuously and concavely on
\([\mu-b,1-b]\), remains constant on \([1-b,1)\), and has an atom of size
\(e^{-\frac{1-\mu}{b}}\) at \(1\). By contrast, \(R^*\) has an atom of size
\(\tfrac12e^{-\frac{1-\mu}{b}}\) at \(\mu-b\), rises continuously and
convexly on \((\mu-b,1-b)\), and has an atom of size \(1/2\) at \(1-b\).

These shapes encode the two candidate equalizers. First,
\eqref{eq:phi-under-Gstar} shows that every deterministic cap in
\([\mu-b,1-b]\) yields \(V^*\) under \(G^*\). By linearity, the same is true
of every random cap supported on this interval, including \(\pi^*\). Second,
Lemma~\ref{lem:pistar-const} shows that all mean-\(\mu\) distributions
supported on \([\mu-b,1-b]\cup\{1\}\), including \(G^*\), yield the same
payoff under \(\pi^*\). Since \(\pi^*\) yields \(V^*\) under \(G^*\), this
common payoff is \(V^*\).

These two equalizing identities identify \((G^*,\pi^*)\) as the candidate
saddle-point pair. They do not by themselves establish global optimality:
Section~\ref{sec:main} proves that no incentive-compatible mechanism can yield
more than \(V^*\) under \(G^*\), and that no distribution in \(\mathcal G_\mu\)
can reduce the payoff of \(\pi^*\) below \(V^*\). Those two bounds complete
the saddle-point argument.

\begin{figure}[t]
\centering
\includegraphics[width=0.8\linewidth]{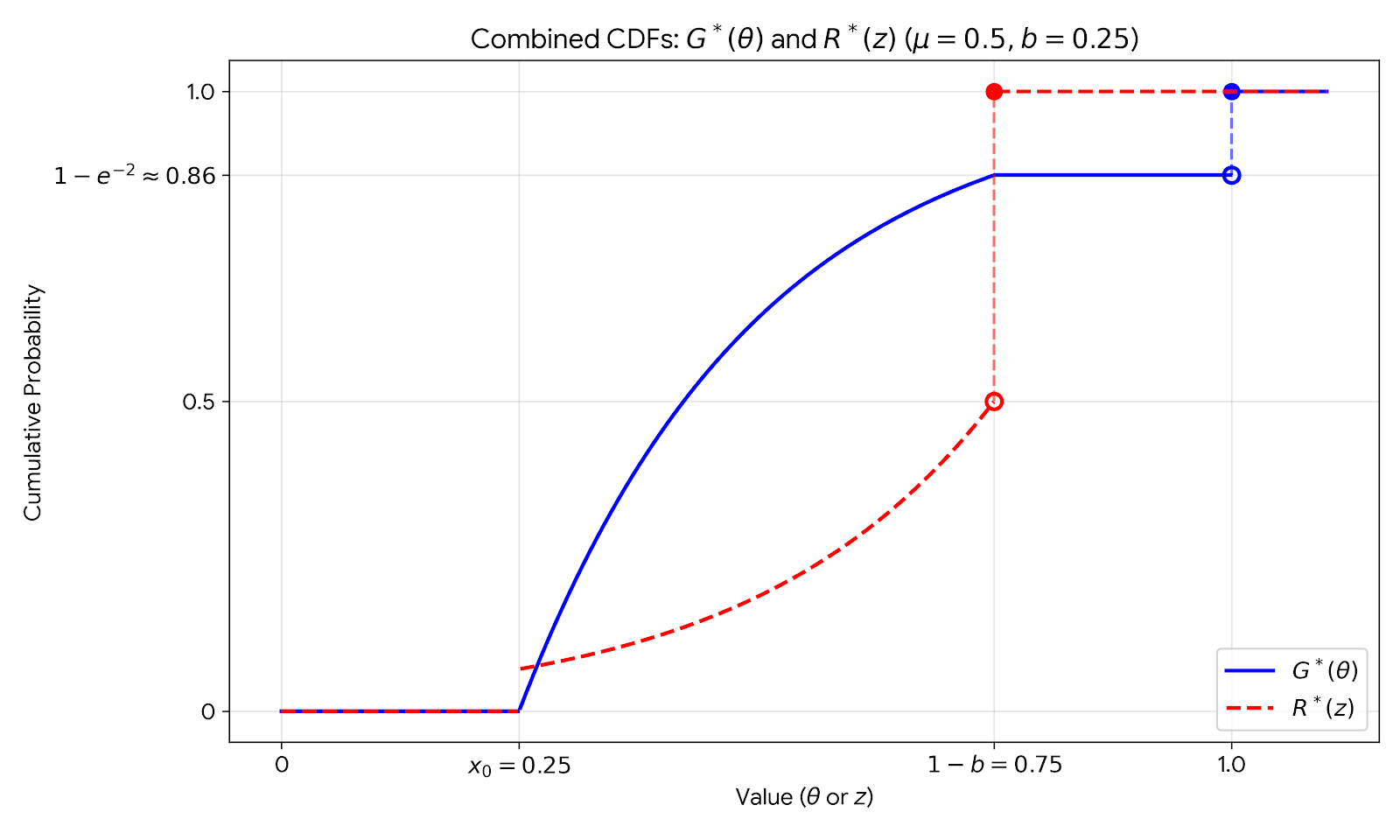}
\caption{The candidate equalizers \(G^*\) and \(R^*\)}
\label{fig:saddle}
\end{figure}

\section{Main result: a saddle point for robust delegation}
\label{sec:main}

We now verify that the candidates constructed in
Section~\ref{sec:heuristic} form a saddle point for the unrestricted
incentive-compatible problem. Let \(G^*\) be the shifted exponential-tail
distribution defined in \eqref{eq:Gstar-x0}, and let \(\pi^*\) be the
mechanism defined in \eqref{eq:pistar}. For any \(\pi\in\Pi^{IC}\) and
\(G\in\mathcal G_\mu\), define the principal's ex ante payoff by
\begin{equation}\label{eq:W-def}
W(\pi,G)
:=\int_{[0,1]}V^\pi(\theta)\,dG(\theta)
=\int_{[0,1]}\int_A v(a,\theta)\,d\pi(a\mid\theta)\,dG(\theta).
\end{equation}
For brevity, write
\[
\kappa^*:=e^{-\frac{1-\mu}{b}}.
\]
Recall from \eqref{eq:V-underbar} that the candidate value is
\begin{equation}\label{eq:Vstar-main}
V^*
:=-\operatorname{Var}_{G^*}(\theta)
=-b^2(1-\kappa^*).
\end{equation}

\begin{theorem}[Saddle point]\label{thm:saddle}
Under Assumption~\ref{ass:mu-range}, the pair \((\pi^*,G^*)\) satisfies
\begin{equation}\label{eq:saddle-ineq}
W(\pi^*,G)
\geq W(\pi^*,G^*)
\geq W(\pi,G^*)
\qquad
\text{for every }\pi\in\Pi^{IC}\text{ and }G\in\mathcal G_\mu.
\end{equation}
Consequently,
\begin{equation}\label{eq:val}
\sup_{\pi\in\Pi^{IC}}\inf_{G\in\mathcal G_\mu}W(\pi,G)
=\inf_{G\in\mathcal G_\mu}\sup_{\pi\in\Pi^{IC}}W(\pi,G)
=V^*.
\end{equation}
Moreover, the robustly optimal mechanism can be implemented as a random
cap with cap distribution \(R^*\) defined in \eqref{eq:Rstar}.
\end{theorem}

\paragraph{Roadmap}
The proof proceeds in two steps. First, we show that \(G^*\) limits the payoff
of every incentive-compatible mechanism to at most \(V^*\). Second, we show
that \(\pi^*\) guarantees at least \(V^*\) against every distribution in
\(\mathcal G_\mu\).

The variance representation
\(V^*=-\operatorname{Var}_{G^*}(\theta)\) also highlights the distinction
between one-sided best responses and saddle-point strategies. Facing the
worst-case shifted exponential-tail distribution \(G^*\), every deterministic
cap in \([\mu-b,1-b]\) attains \(V^*\). Conversely, facing the robustly optimal
mechanism \(\pi^*\), every distribution supported on
\([\mu-b,1-b]\cup\{1\}\) with mean \(\mu\) attains the same value. When
\(\mu\leq1-b\), this class includes the degenerate distribution
\(\delta_\mu\). These simple one-sided best responses do not themselves form
a saddle point. In particular, once the distribution is known to be
\(\delta_\mu\), the principal can implement the constant action \(\mu-b\) and
obtain zero, which strictly exceeds \(V^*\). Likewise, although deterministic
caps are best responses to \(G^*\), they do not guarantee \(V^*\) uniformly
over \(\mathcal G_\mu\). Within the class of cap mechanisms, a saddle point
therefore requires dispersion on both sides: a nondegenerate worst-case
distribution and a lottery over caps. We return to this point in
Section~\ref{subsec:value-randomness}, where we show that every deterministic
cap yields a strictly lower robust payoff than the optimal random cap.

\subsection{Step 1: optimality against
\texorpdfstring{\(G^*\)}{G*}}
\label{subsec:step1}

We begin with a payoff decomposition that applies to every
incentive-compatible mechanism.

\begin{lemma}[Payoff decomposition]\label{lem:EV-formula}
Fix \(x_0\in[0,1]\), let \(G\in\Delta([0,1])\) satisfy
\(\operatorname{supp}(G)\subseteq[x_0,1]\),
and fix \(\pi\in\Pi^{IC}\). Define the signed measure on
\([x_0,1]\) by
\begin{equation}\label{eq:nu-def}
\nu_G(dx)
:=[1-G(x)]\,dx-b\,dG(x).
\end{equation}
Then
\begin{equation}\label{eq:EV-general-x0}
\begin{aligned}
W(\pi,G)
={}&U^\pi(x_0)+x_0^2
+2\int_{[x_0,1]}m^\pi(x)\,\nu_G(dx)\\
&\quad-\mathbb E_G\!\left[(\theta-b)^2\right].
\end{aligned}
\end{equation}
\end{lemma}

The proof is given in Appendix~\ref{app:payoff-decomposition-proof}.

In the spirit of the payoff decomposition in
\citet{kovavc2009stochastic}, Lemma~\ref{lem:EV-formula} rewrites the
principal's ex ante payoff in terms of the agent's interim utility and expected
action. The signed measure \(\nu_{G^*}\) is the measure-valued analogue of the
virtual coefficient multiplying the expected action in their reduction. Our
argument then exploits the additional structure of the equalizing distribution
\(G^*\). In particular, \(\nu_{G^*}\) vanishes on the principal's
indifference region. Only the upper-tail term remains, where \(\nu_{G^*}\)
consists of a positive Lebesgue component on \([1-b,1)\) and a negative atom at
\(1\). This cancellation turns the general payoff decomposition into the sharp
upper bound below.

\begin{proposition}[Optimality against \(G^*\)]
\label{prop:interval-optimal-gstar}
For every \(\pi\in\Pi^{IC}\),
\[
W(\pi,G^*)\leq-\operatorname{Var}_{G^*}(\theta)=V^*.
\]
Equality is attained by every deterministic cap in \([\mu-b,1-b]\) and by
every random cap supported on that interval.
\end{proposition}

\begin{proof}
Substituting \(G^*\) into \eqref{eq:nu-def} gives
\begin{equation}\label{eq:nu-Gstar}
\nu_{G^*}(dx)
=\kappa^*\mathbf 1_{[1-b,1)}(x)\,dx
-b\kappa^*\delta_1(dx).
\end{equation}
Indeed, the absolutely continuous components of
\([1-G^*(x)]\,dx\) and \(b\,dG^*(x)\) cancel on
\([\mu-b,1-b)\); above \(1-b\), the survival function equals \(\kappa^*\),
and \(G^*\) has an atom of size \(\kappa^*\) at \(1\). It follows from
Lemma~\ref{lem:EV-formula} that
\begin{equation}\label{eq:EV-Gstar-tail}
\begin{aligned}
W(\pi,G^*)
={}&U^\pi(\mu-b)+(\mu-b)^2\\
&+2\kappa^*
\left[
\int_{1-b}^1m^\pi(x)\,dx-bm^\pi(1)
\right]
-\mathbb E_{G^*}\!\left[(\theta-b)^2\right].
\end{aligned}
\end{equation}
Incentive compatibility implies that \(m^\pi\) is nondecreasing. Therefore,
\[
\int_{1-b}^1m^\pi(x)\,dx-bm^\pi(1)
=\int_{1-b}^1\bigl[m^\pi(x)-m^\pi(1)\bigr]\,dx
\leq0.
\]
In addition, \(U^\pi(\mu-b)\leq0\), because the agent's payoff is the
negative of a quadratic loss. Hence
\begin{align}
W(\pi,G^*)
&\leq(\mu-b)^2
-\mathbb E_{G^*}\!\left[(\theta-b)^2\right]\nonumber\\
&=-\operatorname{Var}_{G^*}(\theta)
=V^*,
\label{eq:Gstar-upper-bound}
\end{align}
where the second line uses \(\mathbb E_{G^*}[\theta]=\mu\).

Both inequalities are equalities whenever
\[
U^\pi(\mu-b)=0
\quad\text{and}\quad
m^\pi(x)=m^\pi(1)
\quad\text{for almost every }x\in[1-b,1].
\]
Every random cap supported on \([\mu-b,1-b]\) satisfies these conditions:
it implements the action \(\mu-b\) at state \(\mu-b\), and its expected
action is constant for all states weakly above \(1-b\). Deterministic caps in
the same interval are special cases.
\end{proof}

The variance representation in \eqref{eq:Vstar-main} has a simple
interpretation. Under the cap \(z=\mu-b\), the support of \(G^*\) lies weakly
above the cap, so the induced action is constant at \(\mu-b\). The
principal's realized loss is therefore
\[
\bigl((\mu-b)-(\theta-b)\bigr)^2=(\theta-\mu)^2,
\]
and its expected payoff is \(-\operatorname{Var}_{G^*}(\theta)=V^*\). The
equalizing property of \(G^*\) extends the same payoff to every cap in
\([\mu-b,1-b]\).

This result may be of independent interest because the optimality of a
single interval is not automatic when the state distribution has atoms or
lacks full support, as the example above illustrates. Standard
interval-optimality results commonly impose density and support conditions
\citep{alonso2008optimal,amador2013theory}. By contrast, \(G^*\) has both a
zero-density region and an atom, so proving that a single interval remains
optimal is nontrivial.

\subsection{Step 2: worst cases against
\texorpdfstring{\(\pi^*\)}{pi*}}
\label{subsec:step2}

We next characterize all of nature's best responses to \(\pi^*\).

\begin{proposition}[Worst-case distributions under \(\pi^*\)]
\label{prop:Gmin-under-pistar}
Define
\[
\mathcal G_\mu^*
:=\left\{
G\in\mathcal G_\mu:
\operatorname{supp}(G)
\subseteq[\mu-b,1-b]\cup\{1\}
\right\}.
\]
Then
\begin{equation}\label{eq:Gstar-argmin}
\operatorname*{argmin}_{G\in\mathcal G_\mu}W(\pi^*,G)
=\mathcal G_\mu^*,
\qquad
\min_{G\in\mathcal G_\mu}W(\pi^*,G)=V^*.
\end{equation}
\end{proposition}

\begin{proof}
Consider the affine function
\begin{equation}\label{eq:ell-main}
\ell(\theta)
:=-b^2+b\kappa^*\bigl(\theta-(\mu-b)\bigr).
\end{equation}
Using the interim payoff induced by \(\pi^*\), a direct calculation gives
\begin{equation}\label{eq:Vminusell}
V^{\pi^*}(\theta)-\ell(\theta)=
\begin{cases}
b\kappa^*(\mu-b-\theta),
&0\leq\theta<\mu-b,\\[4pt]
0,
&\mu-b\leq\theta\leq1-b,\\[4pt]
(\theta-(1-b))(1-\theta),
&1-b<\theta\leq1.
\end{cases}
\end{equation}
Thus, \(V^{\pi^*}(\theta)\geq\ell(\theta)\) on \([0,1]\), with equality
exactly on \([\mu-b,1-b]\cup\{1\}\). For every \(G\in\mathcal G_\mu\),
\begin{align}
W(\pi^*,G)
&\geq\int_{[0,1]}\ell(\theta)\,dG(\theta)\nonumber\\
&=-b^2+b\kappa^*\bigl(\mu-(\mu-b)\bigr)\nonumber\\
&=-b^2+b^2\kappa^*
=V^*.
\label{eq:pistar-lower-bound}
\end{align}
Equality holds if and only if \(G\) is concentrated on the contact set
\([\mu-b,1-b]\cup\{1\}\), which is equivalent to \(G\in\mathcal G_\mu^*\).
Because \(G^*\in\mathcal G_\mu^*\), the lower bound is attained.
\end{proof}

\begin{proof}[Proof of Theorem~\ref{thm:saddle}]
Proposition~\ref{prop:interval-optimal-gstar} gives
\[
W(\pi,G^*)\leq V^*
\qquad
\text{for every }\pi\in\Pi^{IC},
\]
whereas Proposition~\ref{prop:Gmin-under-pistar} gives
\[
W(\pi^*,G)\geq V^*
\qquad
\text{for every }G\in\mathcal G_\mu.
\]
Applying the first inequality to \(\pi^*\) and the second to \(G^*\) yields
\(W(\pi^*,G^*)=V^*\). The saddle inequalities in
\eqref{eq:saddle-ineq} follow immediately. They imply both sides of the value
identity \eqref{eq:val}, and \(\pi^*\) and \(G^*\) attain the respective
outer extrema.

Finally, \eqref{eq:Rstar} gives the stated random-cap implementation of
\(\pi^*\).
\end{proof}

{
\subsection{Monotone comparative statics of the random cap}
\label{subsec:random-cap-comparative-statics}

Let $Z^*(\mu,b)$ denote a cap drawn from the optimal cdf $R^*$ in
\eqref{eq:Rstar}. Integrating its survival function gives
\begin{equation}
\mathbb E[Z^*(\mu,b)]
=1-\frac{3}{2}b+\frac{b}{2}e^{-(1-\mu)/b}.
\label{eq:mean-optimal-cap}
\end{equation}
Its variance is
\begin{equation}
\Var\bigl(Z^*(\mu,b)\bigr)
=b^2\left[
\frac34-\left(\frac{1-\mu}{b}+\frac12\right)e^{-(1-\mu)/b}
-\frac14e^{-2(1-\mu)/b}
\right].
\label{eq:variance-optimal-cap}
\end{equation}
Direct differentiation gives
\[
\frac{\partial}{\partial\mu}\Var\bigl(Z^*(\mu,b)\bigr)<0,
\qquad
\frac{\partial}{\partial b}\Var\bigl(Z^*(\mu,b)\bigr)>0.
\]
Figure~\ref{fig:random-cap-comparative-statics} illustrates the associated
first-order stochastic shifts in the cap distribution.

\begin{figure}[htbp]
\centering
\includegraphics[width=\linewidth]{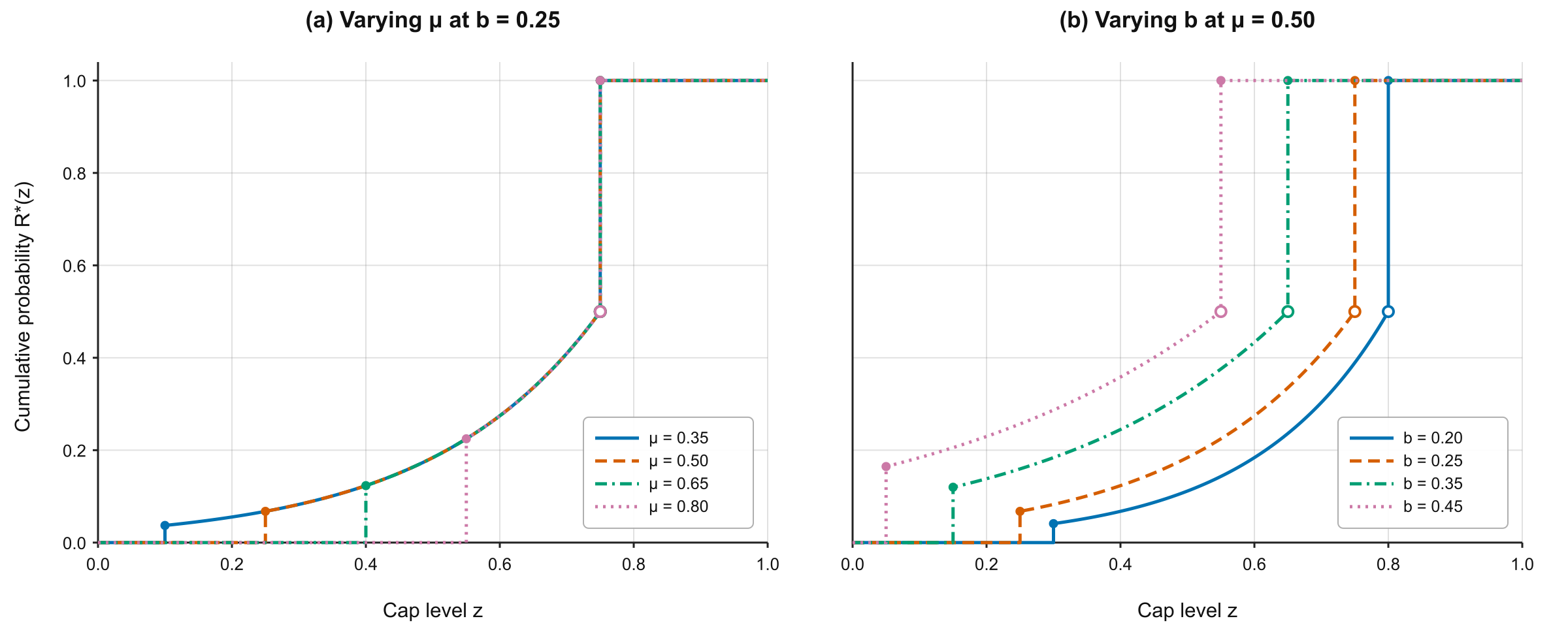}
\caption{Monotone comparative statics of the optimal cap cdf $R^*$. Left:
varying $\mu$ at $b=0.25$. Right: varying $b$ at $\mu=0.50$. Vertical
segments display the endpoint atoms.}
\label{fig:random-cap-comparative-statics}
\end{figure}

\paragraph{The mean $\mu$.}
Holding $b$ fixed, a higher mean shifts $Z^*(\mu,b)$ upward in the
first-order stochastic-dominance order and reduces its variance. Moreover,
\[
\frac{\partial}{\partial\mu}\mathbb E[Z^*(\mu,b)]
=\frac12e^{-(1-\mu)/b}>0.
\]
At the allocation level, the full-delegation region $[0,\mu-b]$ expands, the
stochastic-intervention region $(\mu-b,1-b)$ contracts from the left, and the
always-capped region $(1-b,1]$ is unchanged. For states that remain in the
common stochastic region, the probability that the cap binds is unchanged
because the interior cdf does not depend on $\mu$; intervention nevertheless
becomes less severe because the lowest caps have been raised. Under the
monotone coupling, $a=\min\{\theta,Z^*(\mu,b)\}$ therefore rises state by state.
Economically, a higher mean makes adverse low states less important and allows
the principal to enlarge the safe-harbor region without changing the local
form of discipline. In this sense, $\mu$ is primarily a boundary, or
discretion, parameter. Consistently,
\[
\frac{\partial V^*(\mu,b)}{\partial\mu}
=be^{-(1-\mu)/b}>0,
\]
and $Z^*(\mu,b)$ converges in distribution to the deterministic cap $1-b$ as
$\mu\uparrow1$.

\paragraph{The bias $b$.}
Holding $\mu$ fixed, a larger bias shifts $Z^*(\mu,b)$ downward in the
first-order stochastic-dominance order and increases its variance.
Equation~\eqref{eq:mean-optimal-cap} gives
\[
\frac{\partial}{\partial b}\mathbb E[Z^*(\mu,b)]
=-\frac32+\frac12\left(1+\frac{1-\mu}{b}\right)
e^{-(1-\mu)/b}
\in\left(-\frac32,-1\right).
\]
The full-delegation region $[0,\mu-b]$ contracts, the stochastic-intervention
region $(\mu-b,1-b)$ shifts to the left without changing its length, and the
always-capped region $(1-b,1]$ expands downward. For states that remain in the
common stochastic region, the cap binds more frequently and, conditional on
intervention, the realized cap is lower. Thus, a larger bias tightens the
mechanism along both the extensive margin---the set and probability of
intervention---and the intensive margin---the severity of the restriction.
Economically, stronger preference conflict makes high actions more costly and
calls for a globally tighter system of discipline. Accordingly, $b$ is a
global discipline parameter, and
\[
\frac{\partial V^*(\mu,b)}{\partial b}
=b\left[\left(2+\frac{1-\mu}{b}\right)e^{-(1-\mu)/b}-2\right]<0.
\]
}

\section{Extensions}\label{sec:extensions}
This section develops three extensions of the baseline random-cap analysis.
First, we replace the mean-only ambiguity set with
mean-preserving-contraction (MPC) ambiguity and provide conditions under which
random-cap optimality continues to hold. Second, we solve the robust problem
within the class of deterministic mechanisms to quantify the value of
randomization. Finally, we extend the saddle-point construction to the boundary
regime \(\mu<b\) and study the benchmark in which nature's mean is unrestricted.
Together, these extensions clarify both the scope and the limits of random-cap
optimality.

\subsection{Mean-preserving-contraction (MPC) ambiguity}
\label{sec:MPC}

Convex-order ambiguity provides a natural refinement of the mean-only
ambiguity set studied above. Fix a prior \(F\in\Delta([0,1])\) with mean
\(\mu\), and write \(G\preceq_{\mathrm{cx}}F\) when \(G\) is dominated by
\(F\) in convex order, or equivalently, when \(G\) is a mean-preserving
contraction of \(F\). The MPC problem is
\begin{equation}\label{eq:MPCP}
\sup_{\pi\in\Pi^{IC}}
\inf_{G\preceq_{\mathrm{cx}}F}W(\pi,G),
\qquad
W(\pi,G)
:=-\mathbb E_{\substack{\theta\sim G\\
a\sim\pi(\cdot\mid\theta)}}
\bigl[(\theta-b-a)^2\bigr].
\end{equation}

We impose the following regularity condition on the prior.

\begin{assumption}[Regularity]\label{ass:regu}
The distribution \(F\) is absolutely continuous on \([0,1]\), has full
support, and admits a continuous and strictly positive density \(f\) on
\((0,1)\). Its hazard rate
\[
h_F(x):=\frac{f(x)}{1-F(x)}
\]
is strictly increasing on \((0,1)\).
\end{assumption}

The absolute-continuity requirement is useful for stating the construction
without qualifications concerning atoms. \rev{The Bernoulli prior
\(F^B=(1-\mu)\delta_0+\mu\delta_1\) does not satisfy
Assumption~\ref{ass:regu}, but it provides a useful boundary comparison: its
convex-order lower set consists of all distributions on \([0,1]\) with mean
\(\mu\), so it recovers the mean-only ambiguity set of the baseline
model.\footnote{\rev{This recovery is a separate atomic construction, not an
application of Lemma~\ref{lem:MPC-feasibility}. Let
\(\kappa_B:=e^{-(1-\mu)/b}\) and set
\(\bar z_B=1-b\), \(\tau_B=1\), and \(\eta_B=1/2\). The cutoff is reached
through the endpoint jump:
\(1-F^B(1)=0<\kappa_B<1-F^B(1-)=\mu\), while
\(\tau_B-\bar z_B=b\). Thus, the contact equality and strict gap in
\eqref{eq:MPC-cutoff-properties} do not hold. The resulting distribution is
\(G^*\) in \eqref{eq:Gstar-x0}, with an endpoint atom of size \(\kappa_B\) at
one, and \(\eta_B=1/2\) yields the cap distribution \(R^*\) in
\eqref{eq:Rstar}. The saddle-point conclusion follows from
Theorem~\ref{thm:saddle}, rather than from
Theorem~\ref{thm:MPC-saddle}.}}}

For \(z\in[\mu-b,1]\), define
\begin{equation}\label{eq:H-def}
H_F(z)
:=\int_z^1
\min\!\left\{1,
e^{\frac{z-(\mu-b)}{b}}[1-F(x)]\right\}dx-b.
\end{equation}
The next lemma constructs the candidate worst-case distribution.

\begin{lemma}[MPC benchmark distribution]\label{lem:MPC-feasibility}
Under Assumptions~\ref{ass:mu-range} and~\ref{ass:regu}, there is a unique
\(\bar z_F\in(\mu-b,1)\) satisfying
\begin{equation}\label{eq:zbar-def}
H_F(\bar z_F)=0.
\end{equation}
Define
\begin{equation}\label{eq:tau-def}
\tau_F
:=\inf\left\{x\in[\bar z_F,1]:
1-F(x)\leq
e^{-\frac{\bar z_F-(\mu-b)}{b}}
\right\}.
\end{equation}
Then
\begin{equation}\label{eq:MPC-cutoff-properties}
\bar z_F<\tau_F<1,
\qquad
1-F(\tau_F)
=e^{-\frac{\bar z_F-(\mu-b)}{b}},
\qquad
0<\tau_F-\bar z_F<b.
\end{equation}
Moreover, the distribution \(G^F\) with survival function
\begin{equation}\label{eq:GF-def}
1-G^F(x)=
\begin{cases}
1,
&x<\mu-b,\\[4pt]
e^{-\frac{x-(\mu-b)}{b}},
&\mu-b\leq x<\bar z_F,\\[9pt]
e^{-\frac{\bar z_F-(\mu-b)}{b}},
&\bar z_F\leq x<\tau_F,\\[4pt]
1-F(x),
&\tau_F\leq x\leq1
\end{cases}
\end{equation}
has mean \(\mu\) and satisfies \(G^F\preceq_{\mathrm{cx}}F\).
\end{lemma}

The construction combines three pieces. Starting at \(\mu-b\), the survival
function falls exponentially until \(\bar z_F\), remains constant until
\(\tau_F\), and then coincides with the prior's survival function. The root
condition \eqref{eq:zbar-def} makes these pieces satisfy the mean constraint;
the increasing-hazard-rate assumption ensures both uniqueness and convex-order
feasibility.

Define
\begin{equation}\label{eq:etaF-def}
\eta_F
:=1-\frac{\tau_F-\bar z_F}{2b}.
\end{equation}
By \eqref{eq:MPC-cutoff-properties}, \(\eta_F\in(1/2,1)\). Consider the
direct mechanism whose conditional cdf is
\begin{equation}\label{eq:piF}
\pi^F(a\mid\theta)=
\begin{cases}
0,
&a<\min\{\theta,\mu-b\},\\[6pt]
\eta_F e^{\frac{a-\bar z_F}{b}},
&\mu-b\leq a<\min\{\theta,\bar z_F\},\\[9pt]
1,
&a\geq\min\{\theta,\bar z_F\}.
\end{cases}
\end{equation}

\begin{theorem}[Robust random cap under MPC ambiguity]
\label{thm:MPC-saddle}
Under Assumptions~\ref{ass:mu-range} and~\ref{ass:regu}, the pair
\((\pi^F,G^F)\) is a saddle point of \eqref{eq:MPCP}: for every
\(\pi\in\Pi^{IC}\) and every \(G\preceq_{\mathrm{cx}}F\),
\begin{equation}\label{eq:MPC-saddle-ineq}
W(\pi^F,G)
\geq W(\pi^F,G^F)
\geq W(\pi,G^F).
\end{equation}
Consequently, the value of \eqref{eq:MPCP} is
\begin{equation}\label{eq:MPC-value}
W(\pi^F,G^F)=-\operatorname{Var}_{G^F}(\theta).
\end{equation}
The mechanism \(\pi^F\) is implemented by drawing a cap \(Z\) from the cdf
\begin{equation}\label{eq:mpc-Rstar}
R^F(z)=
\begin{cases}
0,
&z<\mu-b,\\[4pt]
\eta_F e^{\frac{z-\bar z_F}{b}},
&\mu-b\leq z<\bar z_F,\\[9pt]
1,
&z\geq\bar z_F,
\end{cases}
\end{equation}
and setting \(a=\min\{\theta,Z\}\).
\end{theorem}

The proof is given in Appendix~\ref{app:MPC-proofs}. The upper saddle bound
extends the signed-measure payoff decomposition used in the baseline model.
The lower bound combines an affine function with a convex correction, thereby
using exactly the additional restriction imposed by convex-order dominance.

\subsubsection{A signal-design interpretation}
\label{subsec:MPC-signal-foundation}

Mean-preserving contractions also arise naturally in information design
\citep{kamenica2011bayesian,gentzkow2016rothschild,chen2023information}. The
connection to the direct MPC benchmark, however, is only partial. When the
relevant type is a posterior mean, residual uncertainty about the latent state
remains payoff relevant and introduces an additional variance term into the
principal's objective.

Let the latent state \(\theta\) have prior \(F\). A signal realization \(s\)
induces a posterior \(P_s\in\Delta([0,1])\), with posterior mean and variance
\[
x_s:=\mathbb E_{P_s}[\theta],
\qquad
\sigma_s^2:=\operatorname{Var}_{P_s}(\theta).
\]
Conditional on \(s\), the principal's and the agent's expected payoffs from
action \(a\) are
\begin{equation}\label{eq:posterior-payoffs}
\begin{aligned}
v_s(a)&=-\sigma_s^2-(x_s-b-a)^2,\\
u_s(a)&=-\sigma_s^2-(x_s-a)^2.
\end{aligned}
\end{equation}
{
The posterior-variance term is independent of the action. It therefore drops
out of the agent's reporting incentives but remains part of the principal's ex
ante payoff. Consequently, posteriors with the same mean induce identical
preferences over action lotteries. Under weak incentive compatibility,
however, this observation does not require a mechanism to assign the same
lottery to all such posteriors: different lotteries may give the agent the same
expected quadratic loss while giving the principal different payoffs.

Let \(\tau\in\Delta(\Delta([0,1]))\) denote the distribution of posteriors
induced by a signal structure. Bayesian plausibility requires
\begin{equation}\label{eq:Bayes-plausibility}
\int_{\Delta([0,1])}P\,d\tau(P)=F.
\end{equation}
For a posterior \(P\), write \(x_P:=\mathbb E_P[\theta]\). An posterior-report mechanism \(\widetilde\pi(\cdot\mid P)\) assigns a lottery
over actions to every reported posterior. It is incentive compatible if
\begin{equation}\label{eq:signal-IC}
\int_A(x_P-a)^2\,d\widetilde\pi(a\mid P)
\leq
\int_A(x_P-a)^2\,d\widetilde\pi(a\mid\widehat P)
\qquad
\forall P,\widehat P\in\Delta([0,1]).
\end{equation}
Let \(\widetilde\Pi^{IC}\) denote this mechanism class.

The adversarial signal-design problem is
\begin{equation}\label{eq:signal-maxmin}
\sup_{\widetilde\pi\in\widetilde\Pi^{IC}}
\inf_{\substack{\tau\in\Delta(\Delta([0,1])):\\
\int P\,d\tau(P)=F}}
\mathbb E_{\substack{P_s\sim\tau\\
a\sim\widetilde\pi(\cdot\mid P_s)}}
\bigl[v_s(a)\bigr].
\end{equation}

\begin{theorem}[Signal-design reduction and robust optimum]
\label{thm:MPC-signal}
Suppose that \(F\in\Delta([0,1])\) has mean \(\mu\) and that
\(0<b<\mu<1\). For \(\pi\in\Pi^{IC}\) and
\(\bar G\preceq_{\mathrm{cx}}F\), define
\begin{equation}\label{eq:signal-objective-decomposition}
\mathcal J(\pi,\bar G)
:=-\operatorname{Var}_F(\theta)
+\operatorname{Var}_{\bar G}(x)
-\mathbb E_{\substack{x\sim\bar G\\
a\sim\pi(\cdot\mid x)}}
\left[(x-b-a)^2\right].
\end{equation}
Then the signal-design problem \eqref{eq:signal-maxmin} is
value-equivalent to
\begin{equation}\label{eq:signal-MPC}
\sup_{\pi\in\Pi^{IC}}
\inf_{\bar G\preceq_{\mathrm{cx}}F}
\mathcal J(\pi,\bar G).
\end{equation}
Their common max--min value is
\[
-\operatorname{Var}_F(\theta).
\]
The fully uninformative signal \(\tau^{\mathrm{FP}}=\delta_F\), which induces
the distribution \(\bar G=\delta_\mu\) of posterior means, is a worst-case
signal for both mechanisms described below. The singleton delegation rule
\(a=\mu-b\) and the posterior-mean cap
\(a(P)=\min\{x_P,\mu-b\}\) are both robustly optimal within \(\widetilde\Pi^{IC}\).
\end{theorem}

\noindent
The proof is provided in Appendix~\ref{app:proof-MPC-signal}. Full pooling gives
the same upper bound in the signal-design problem and in
\eqref{eq:signal-MPC}, while the constant action \(a=\mu-b\) attains that bound
under every signal structure. Thus, allowing the mechanism to condition on the
full reported posterior rather than only its mean creates no additional robust
value.

At the level of payoffs, signal design remains distinct from the direct MPC
problem. In the latter, the realized type is itself the primitive
payoff-relevant state. Under signal design, residual posterior variance enters
the principal's objective. The adversarial information designer can choose a
fully uninformative signal and eliminate the agent's informational advantage.
In the baseline model, by contrast, the agent always observes the state
perfectly, and nature varies only its distribution. Full pooling therefore
shuts down the commitment--flexibility tradeoff at the heart of delegation.
}

\subsection{The value of randomization}
\label{subsec:value-randomness}

\rev{\citet{HuLi2023} show that, under their compact-action-space conditions,
an optimal deterministic delegation set in the mean-constrained robust problem
can be chosen to be an interval. Our quadratic environment satisfies the
payoff condition behind their lower-convex-envelope argument: under full
discretion, the principal's payoff is concave in the state. After normalizing
intervals,\footnote{\rev{Our maintained action space $A=\mathbb R$ lies outside
their formal assumptions, but the short scope check in
Appendix~\ref{app:deterministic} shows that this creates no value gap.}} a
robustly optimal deterministic mechanism can therefore be chosen to be a cap
$[0,z]$.}

Let
$V_z^D(\mu):=\inf_{G\in\mathcal G_\mu}\phi_z(G)$ denote the robust payoff
generated by the deterministic cap $z$. Appendix~\ref{app:deterministic}
solves this benchmark in two steps. For a fixed cap, nature's problem admits an
extremal solution supported on at most two states. Optimizing the resulting
payoff over $z$ then gives
\begin{equation}\label{eq:optimal-det-cap-value}
V^D(\mu)
:=
\sup_{z\in[0,1]}V_z^D(\mu)
=
\begin{cases}
-b^2,
& \mu\leq1-2b,\\[6pt]
-\left[
b(1-\mu)-\dfrac{(1-\mu)^2}{4}
\right],
& \mu>1-2b.
\end{cases}
\end{equation}
If $\mu\leq1-2b$, every cap $z\in[1-2b,1]$ is optimal. If
$\mu>1-2b$, the optimal deterministic cap is unique and equals
\begin{equation}\label{eq:optimal-det-cap}
z_D^*
=
\frac{1+\mu}{2}-b.
\end{equation}

Two observations are worth emphasizing. \rev{A worst-case distribution can
always be chosen with support on at most two states.}
This is consistent with the extreme-point structure of a linear moment problem
with a single mean restriction. Second, comparing
\eqref{eq:optimal-det-cap-value} with the saddle-point value gives
\[
V^*
=
-b^2\left(
1-e^{-\frac{1-\mu}{b}}
\right)
>
V^D(\mu)
\qquad
\text{for }0<b<\mu<1.
\]
Thus, randomization strictly improves the principal's robust guarantee. A
deterministic cap leaves the principal exposed to an extremal distribution
tailored to that particular cap. A lottery over caps instead hedges across
these cap-specific attacks and reshapes the principal's state-by-state payoff
so that no feasible distribution can lower its expectation below $V^*$.
Randomization is valuable not because it improves performance under the
least-favorable distribution itself, but because it provides protection
against deviations from that distribution. Figure~\ref{fig:value-comparison}
compares the Bayesian value $V^U$ under the uniform prior, the robust
random-cap value $V^*$, and the robust deterministic-cap value $V^D$.

\begin{figure}[htbp]
\centering
\includegraphics[width=0.9\textwidth]{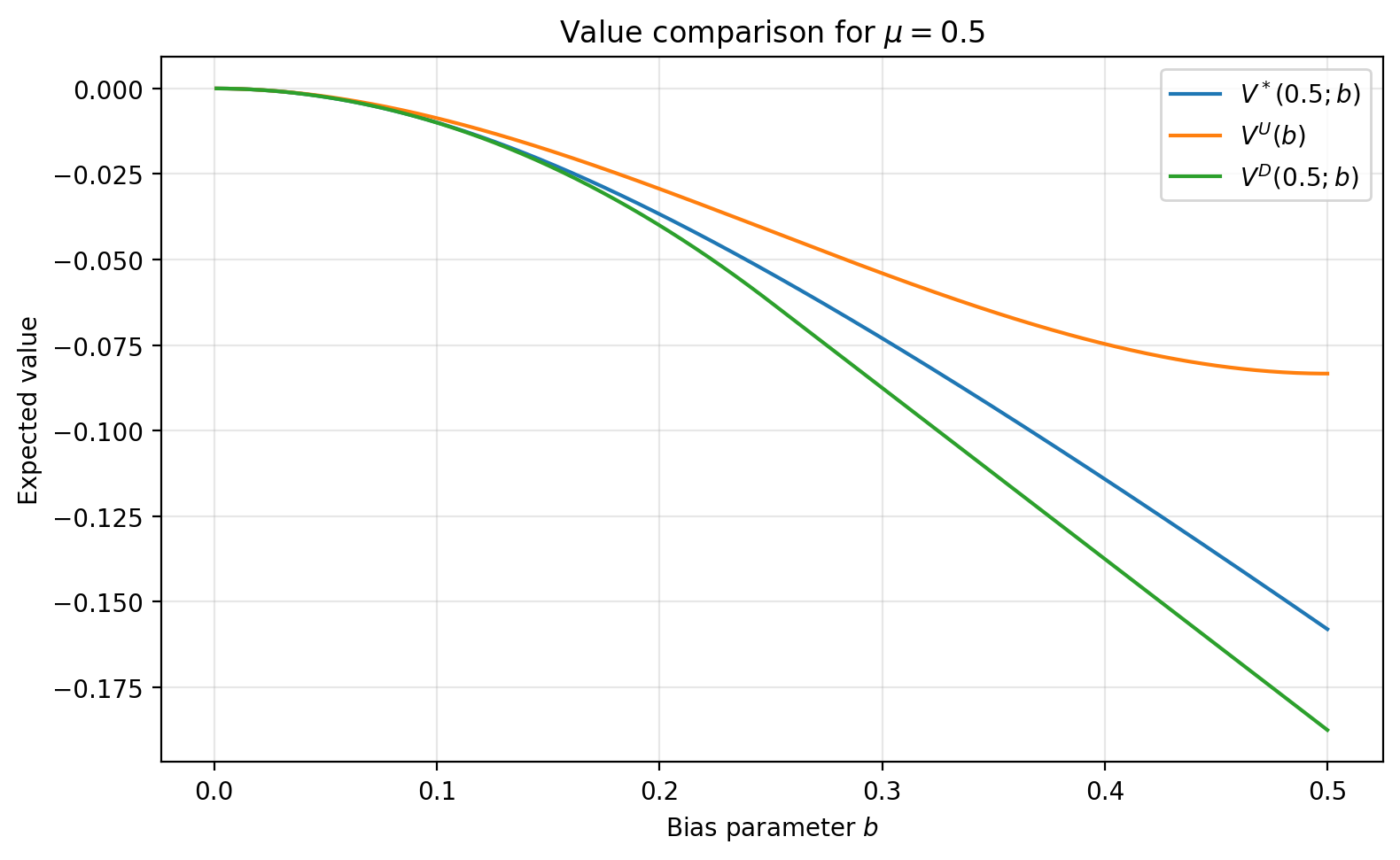}
\caption{Value comparison for $\mu=0.5$: the Bayesian uniform-prior
benchmark $V^U$, the robust random-cap value $V^*$, and the robust
deterministic-cap value $V^D$.}
\label{fig:value-comparison}
\end{figure}

\subsection{The general saddle point and the boundary regime}
\label{subsec:boundary}

The baseline analysis focuses on the interior regime $\mu>b$, in which the
lower endpoint $\mu-b$ of the double-indifference region lies in the state
space. The same saddle-point construction, however, extends to every
$\mu\in[0,1]$.\footnote{For $x\in\mathbb R$, write
$x_+:=\max\{x,0\}$ and $x_-:=\max\{-x,0\}$.}

Let $\pi_\mu^*$ denote the robustly optimal mechanism for any mean. It
retains the random-cap component of the baseline solution and introduces an
additional outcome lottery. More precisely, the mechanism proceeds in three
steps.

\begin{enumerate}
\item[\textbf{Step 1.}] Draw a cap \(Z\sim R_\mu^*\), independently of the
state, where \(R_\mu^*\) is the cdf
\begin{equation}\label{eq:Rstar-all-mu}
R_\mu^*(z)
=
\begin{cases}
0,
&z<(\mu-b)_+,\\[4pt]
\dfrac12 e^{\frac{z-1+b}{b}},
&(\mu-b)_+\leq z<1-b,\\[8pt]
1,
&z\geq1-b.
\end{cases}
\end{equation}

\item[\textbf{Step 2.}] After observing \(Z\), the agent chooses a base action
\(y\in[0,Z]\).

\item[\textbf{Step 3.}] Conditional on \(y\), draw the final action from the
lottery \(A_y^\mu\), defined by\footnote{When \((\mu-b)_-=0\), adopt the
convention \(A_y^\mu=\delta_y\).}
\begin{equation}\label{eq:boundary-A}
A_y^\mu
=
\frac{2y}{2y+(\mu-b)_-}\,\delta_y
+
\frac{(\mu-b)_-}{2y+(\mu-b)_-}\,
\delta_{-\left(y+(\mu-b)_-\right)}.
\end{equation}
\end{enumerate}

This construction has three notable features. First, when $\mu\geq b$, the
negative-part term vanishes. The outcome lottery in
\eqref{eq:boundary-A} becomes degenerate, the lower endpoint of the cap
distribution is $\mu-b$, and $\pi_\mu^*$ reduces exactly to the random-cap
mechanism studied in the baseline analysis.

Second, when $\mu<b$, the additional outcome lottery is essential rather
than merely an alternative implementation of a random cap. At the bottom
state, the agent chooses the base action $y=0$.
\rev{Equation~\eqref{eq:boundary-A} then implies that the final action equals
$\mu-b$ with probability one.} Consequently,
\[
m^{\pi_\mu^*}(0)=\mu-b<0.
\]
Proposition~\ref{prop:imple}, by contrast, requires $m(0)=0$ for every
mechanism that is payoff-equivalent to a random cap. The boundary mechanism
therefore cannot be implemented by a pure random cap.

Third, although the additional lottery changes the implemented action, it
does not change the agent's choice of the base action. It shifts the expected
action downward by $(\mu-b)_-$ while adding to the agent's expected
quadratic loss only a term that is independent of $y$. The agent's ranking
over base actions is therefore unchanged, and she continues to choose
\[
y=\min\{\theta,Z\},
\]
exactly as under an ordinary cap. In this sense, the additional outcome
lottery is incentive neutral.

The least-favorable distribution admits a similarly simple interpretation.
For $\mu\geq b$, it coincides with the baseline shifted exponential-tail
distribution $G^*$. For
$\mu<b$, let $G_b^*$ denote the equalizing distribution obtained by
evaluating the general construction at $\mu=b$. Then
\begin{equation}\label{eq:boundary-G-description}
G_\mu^*
=
\frac{\mu}{b}G_b^*
+
\left(
1-\frac{\mu}{b}
\right)\delta_0.
\end{equation}
Nature therefore retains the same exponential equalizing component but
reduces its weight to $\mu/b$, placing the remaining probability at the
bottom state. This additional atom lowers the mean without disturbing the
conditional-tail equalization that underlies the saddle-point construction.

Appendix~\ref{app:general-saddle} gives the unified formulas and proves
formally that $(\pi_\mu^*,G_\mu^*)$ forms a saddle point for every
$\mu\in[0,1]$. It also shows that pure random caps are strictly suboptimal
when $\mu<b$.

\subsection{The unrestricted-mean benchmark}
\label{subsec:no-mean}

{
The all-means saddle-point construction also allows us to solve the
unrestricted problem in which no mean is specified and nature may choose any
distribution on $[0,1]$.

\begin{corollary}[Unrestricted-mean saddle point]
\label{cor:no-mean}
Define
\[
\mu^\dagger
:=
b\left[
1-\frac12 e^{-\frac{1-b}{b}}
\right]
\in
\left(\frac b2,b\right).
\]
Then $(\pi_{\mu^\dagger}^*,G_{\mu^\dagger}^*)$ is a saddle point on
$\Pi^{IC}\times\Delta([0,1])$: for every $\pi\in\Pi^{IC}$ and
$G\in\Delta([0,1])$,
\begin{equation}
W(\pi,G_{\mu^\dagger}^*)
\leq
W(\pi_{\mu^\dagger}^*,G_{\mu^\dagger}^*)
=
-(\mu^\dagger)^2
\leq
W(\pi_{\mu^\dagger}^*,G).
\label{eq:no-mean-saddle}
\end{equation}
\end{corollary}

The unrestricted-mean problem was introduced in ongoing work by
\citet{HeLiLiYe2026}. Corollary~\ref{cor:no-mean} provides an alternative
derivation from our family of fixed-mean problems. Indeed, the collection
\(
\{\mathcal G_\mu:\mu\in[0,1]\}
\)
partitions $\Delta([0,1])$. Proposition~\ref{prop:general-saddle} therefore
implies
\begin{equation}
\inf_{G\in\Delta([0,1])}
\sup_{\pi\in\Pi^{IC}}W(\pi,G)
=
\min_{\mu\in[0,1]}V^*(\mu).
\label{eq:no-mean-min-value}
\end{equation}
Minimizing \eqref{eq:value-all-mu} gives the unique minimizer $\mu^\dagger$
and the minimum value
$V^*(\mu^\dagger)=-(\mu^\dagger)^2$. The fixed-mean saddle inequality then
gives
\[
W(\pi,G_{\mu^\dagger}^*)
\leq
-(\mu^\dagger)^2
\qquad
\text{for every }\pi\in\Pi^{IC}.
\]

For the reverse saddle inequality, equations~\eqref{eq:affine-all-mu} and
\eqref{eq:global-minorant-all-mu} imply\footnote{\rev{At
$\mu=\mu^\dagger<b$, the slope in \eqref{eq:affine-all-mu} is zero, so
$L_{\mu^\dagger}(\theta)=-(\mu^\dagger)^2$ for every $\theta$.}}
\[
V^{\pi_{\mu^\dagger}^*}(\theta)
\geq
-(\mu^\dagger)^2
\qquad
\text{for every }\theta\in[0,1],
\]
and therefore
\[
W(\pi_{\mu^\dagger}^*,G)
\geq
-(\mu^\dagger)^2
\qquad
\text{for every }G\in\Delta([0,1]).
\]
Together, these two bounds prove \eqref{eq:no-mean-saddle} and imply that the
unrestricted max--min and min--max values both equal $-(\mu^\dagger)^2$.
}

\section{Discussion}\label{sec:discussion}

\subsection{A financial interpretation: debt and costly state verification}
\label{subsec:CSV}

The random-cap characterization admits a natural financial-contracting
interpretation. Let \(\theta\) denote realized project cash flow and let
\(m(\theta)\) denote the lender's expected repayment. For a fixed face value
\(d\), the repayment rule
\[
r_d(\theta)=\min\{\theta,d\}
\]
is the repayment schedule of standard debt. When cash flow is below the
face value, the lender receives the available cash flow; when cash flow is
sufficiently high, repayment is capped at the promised amount \(d\). In
classical models of costly state verification, the same schedule is associated
with a default threshold: low cash-flow realizations trigger verification,
whereas sufficiently high realizations are settled through the fixed promised
payment \citep{townsend1979optimal,gale1985incentive}.

A random cap replaces the fixed face value by a random variable \(Z\), drawn
independently of project cash flow. Conditional on a realization of \(Z\), the
contract remains an ordinary debt contract, and realized repayment is
\[
\min\{\theta,Z\}.
\]
The mechanism can therefore be interpreted as a lottery over standard debt
contracts. This need not correspond to a literal coin flip within a single
loan. The same reduced-form repayment schedule can arise from assigning
different face values across otherwise comparable loans, projects, or
borrowers. What matters for the analogy is that the face-value assignment is
independent of realized cash flow.

At the level of expected repayment, Proposition~\ref{prop:imple} shows that
random caps generate schedules of the form
\begin{equation}\label{eq:debt-mixture-representation}
m(\theta)
=
\mathbb E\bigl[\min\{\theta,Z\}\bigr]
=
\int_{[0,1]}\min\{\theta,d\}\,dR(d).
\end{equation}
The familiar restrictions from security design remain present: repayment
cannot exceed available cash flow, and both the lender's claim and the
borrower's residual claim should be nondecreasing in cash flow
\citep{innes1990limited,Gui2025}. Implementability as a lottery over debt
contracts imposes an additional restriction: the expected repayment schedule
must be concave.

This concavity has a simple economic interpretation. The first dollars of
project cash flow are transferred to the lender under almost every positive
face value. As cash flow rises, an additional dollar is transferred only for
those debt contracts whose face values have not yet been reached.
The fraction of contracts for which this is true declines with cash flow.
Consequently, the lender's marginal repayment rate falls as the project
becomes more successful. Conversely, every repayment schedule satisfying the
conditions in Proposition~\ref{prop:imple} can be decomposed into a unique
distribution of face values. Deterministic debt contracts are therefore the
primitive building blocks of the entire class.

This interpretation also clarifies the sense in which the random-cap mechanism
remains simple. Although its expected repayment schedule may be smooth and
nonlinear, implementation requires randomization over only one parameter---the
face value. Each realization of the mechanism is still a standard debt
contract rather than a bespoke nonlinear security. Randomization changes the
location of the repayment threshold while preserving the familiar debt form
conditional on that threshold.

The connection is nevertheless a reduced-form one. Our model contains neither
an endogenous verification decision nor an audit cost, and it does not model
default reporting, enforcement, or the timing of verification. The results
therefore do not establish that lotteries over debt contracts are optimal in a
costly-state-verification environment. They show instead that random caps and
random face values generate the same state-contingent allocation structure.
Providing a literal financial foundation would require specifying why the
face value is randomized and how that randomization interacts with default,
verification, and ex ante investment incentives.

\subsection{Scope and limitations}
\label{subsec:limitations}

Our results rely on quadratic loss and constant bias, which play distinct
roles. Quadratic loss is primarily a tractability assumption: the consequences
of a stochastic mechanism depend only on the conditional mean and variance of
the action. With nonquadratic preferences, higher moments or the entire action
distribution may matter, making both incentive compatibility and the
adversary's problem substantially more complex.

Constant bias is more economically restrictive. The Bayesian delegation
literature allows more general, state-dependent biases while retaining
quadratic loss
\citep{melumad1991communication,kovavc2009stochastic,kleiner2021extreme}.
An affine specification, \(\beta(\theta)=b_0+b_1\theta\), is a natural first
extension because it allows the conflict of interest to vary while preserving
the ordering of ideal actions.

That ordering alone, however, does not preserve random-cap optimality. The
equalization conditions must generate a valid distribution over cap levels and
a global payoff guarantee. With state-dependent bias, the candidate cap
distribution may fail to be monotone or feasible, and a valid local solution
need not extend to a saddle point.

More generally, a single upper bound may be too restrictive. Robust
implementation may require randomizing over both upper and lower bounds, giving
rise to stochastic intervals or richer delegation sets. The random-cap result
should therefore be viewed as a structural benchmark rather than a universal
claim. A natural direction for future work is to characterize when this simple
random-threshold structure survives under more general preferences.

\clearpage
\appendix

\section{Incentive compatibility and random-cap representability}
\label{app:random-cap}

This appendix relates the random-cap characterization in
Proposition~\ref{prop:imple} to the general characterization of stochastic
incentive-compatible mechanisms in \citet{kovavc2009stochastic}.

\subsection{The Kov\'{a}\v{c}--Mylovanov characterization}

The following lemma restates Lemma~1 of
\citet{kovavc2009stochastic} in our notation. Because the action space is
\(A=\mathbb R\), any Borel pair consisting of a mean and a nonnegative
variance can be implemented by a Borel stochastic kernel, for example by a
two-point lottery.

\begin{lemma}[Kov\'{a}\v{c}--Mylovanov characterization]
\label{lem:IC}
Let \(m,q:[0,1]\to\mathbb R\) be Borel measurable, and define
\[
U(\theta):=-[\theta-m(\theta)]^2-q(\theta).
\]
There exists an incentive-compatible direct mechanism with conditional mean
\(m\) and conditional variance \(q\) if and only if:
\begin{enumerate}
    \item[(i)] \(m\) is nondecreasing;
    \item[(ii)] for every \(\theta\in[0,1]\),
    \begin{equation}\label{eq:KM-variance}
    q(\theta)
    =-U(0)-[m(\theta)-\theta]^2
    -2\int_0^\theta[m(t)-t]\,dt;
    \end{equation}
    \item[(iii)] \(q(\theta)\geq0\) for every \(\theta\in[0,1]\).
\end{enumerate}
Equivalently, condition \textnormal{(ii)} is the envelope identity
\begin{equation}\label{eq:KM-envelope}
U(\theta)
=U(0)+2\int_0^\theta[m(t)-t]\,dt.
\end{equation}
\end{lemma}

When one starts from a direct mechanism rather than from an arbitrary pair
\((m,q)\), condition \textnormal{(iii)} is automatic. Its explicit appearance
in Lemma~\ref{lem:IC} records the feasibility restriction that a conditional
variance must be nonnegative.

\subsection{Proof of Proposition~\ref{prop:imple}}
\label{app:random-cap-proof}

\begin{proof}[Proof of Proposition~\ref{prop:imple}]
For necessity, suppose that \(Z\) is supported on \([0,1]\), and let
\(A_\theta:=\min\{\theta,Z\}\). For every fixed \(z\), the function
\(\theta\mapsto\min\{\theta,z\}\) is nondecreasing and concave, vanishes at
zero, and is bounded above by \(\theta\). Taking expectations proves
conditions \textnormal{(a)} and \textnormal{(b)}.

Let \(S_Z(t):=\Pr(Z>t)\). The layer-cake representation gives
\begin{equation}\label{eq:random-cap-first-moment}
m(\theta)
=\mathbb E[A_\theta]
=\int_0^\theta S_Z(t)\,dt
\end{equation}
and
\begin{align*}
\mathbb E[A_\theta^2]
&=2\int_0^\theta tS_Z(t)\,dt\\
&=2\theta m(\theta)-2\int_0^\theta m(t)\,dt,
\end{align*}
where the second equality follows from integration by parts. Subtracting
\(m(\theta)^2\) yields condition \textnormal{(c)}.

For sufficiency, suppose conditions \textnormal{(a)}--\textnormal{(c)} hold.
Concavity implies that the right derivative \(m'_+\) exists on \([0,1)\), is
nonincreasing and right-continuous, and satisfies
\[
0\leq m'_+(z)\leq m'_+(0)
=\lim_{\theta\downarrow0}\frac{m(\theta)}{\theta}\leq1.
\]
The first inequality follows from monotonicity, and the last follows from
\(m(\theta)\leq\theta\). Therefore,
\[
R(z):=
\begin{cases}
0, & z<0,\\[3pt]
1-m'_+(z), & 0\leq z<1,\\[3pt]
1, & z\geq1
\end{cases}
\]
is a cdf supported on \([0,1]\). Let \(Z\sim R\). Then
\(\Pr(Z>t)=m'_+(t)\) for every \(t\in[0,1)\), so absolute continuity of
\(m\) and condition \textnormal{(a)} imply
\[
\mathbb E[\min\{\theta,Z\}]
=\int_0^\theta m'_+(t)\,dt
=m(\theta).
\]
Moreover,
\begin{align*}
\operatorname{Var}(\min\{\theta,Z\})
&=2\int_0^\theta tm'_+(t)\,dt-m(\theta)^2\\
&=2\theta m(\theta)-2\int_0^\theta m(t)\,dt-m(\theta)^2\\
&=q(\theta),
\end{align*}
where the last equality is condition \textnormal{(c)}. Thus, the random cap
reproduces \((m,q)\).

Finally, \eqref{eq:random-cap-first-moment} implies
\[
1-R(z)=m'_+(z)
\qquad
\text{for every }z\in[0,1).
\]
Hence \(m\) uniquely determines \(R\), including its possible atoms at zero
and one.
\end{proof}

\subsection{Proof of Lemma~\ref{lem:EV-formula}}
\label{app:payoff-decomposition-proof}

\begin{proof}[Proof of Lemma~\ref{lem:EV-formula}]
Incentive compatibility gives the envelope identity
\[
\frac{d}{d\theta}U^\pi(\theta)
=2\bigl(m^\pi(\theta)-\theta\bigr)
\qquad
\text{for almost every }\theta\in[0,1].
\]
Moreover, quadratic preferences imply
\begin{equation}\label{eq:V-U-identity}
V^\pi(\theta)
=
U^\pi(\theta)
+2b\bigl(\theta-m^\pi(\theta)\bigr)-b^2.
\end{equation}

Because \(G\) is supported on \([x_0,1]\), the envelope identity and
Fubini's theorem yield
\[
\mathbb E_G[U^\pi(\theta)]
=
U^\pi(x_0)
+2\int_{x_0}^1
\bigl(m^\pi(x)-x\bigr)[1-G(x)]\,dx.
\]
Substituting this expression into \eqref{eq:V-U-identity} gives
\begin{align*}
W(\pi,G)
={}&U^\pi(x_0)
+2\int_{x_0}^1m^\pi(x)[1-G(x)]\,dx
-2b\int_{[x_0,1]}m^\pi(x)\,dG(x)\\
&\quad
-2\int_{x_0}^1x[1-G(x)]\,dx
+2b\,\mathbb E_G[\theta]-b^2.
\end{align*}
The support restriction also implies
\[
\mathbb E_G[\theta^2]
=
x_0^2
+2\int_{x_0}^1x[1-G(x)]\,dx.
\]
Therefore,
\begin{align*}
W(\pi,G)
={}&U^\pi(x_0)+x_0^2
+2\int_{[x_0,1]}m^\pi(x)
\bigl([1-G(x)]\,dx-b\,dG(x)\bigr)\\
&\quad
-\mathbb E_G[\theta^2]
+2b\,\mathbb E_G[\theta]-b^2\\
={}&U^\pi(x_0)+x_0^2
+2\int_{[x_0,1]}m^\pi(x)\,\nu_G(dx)
-\mathbb E_G[(\theta-b)^2],
\end{align*}
which proves \eqref{eq:EV-general-x0}.
\end{proof}

\subsection{Implications and the incentive-compatible case}

\begin{remark}[Implied properties]\label{rem:implied}
Conditions \textnormal{(a)}--\textnormal{(c)} of
Proposition~\ref{prop:imple} imply \(q\geq0\). To see this, define
\[
\varphi(\theta)
:=2\theta m(\theta)-2\int_0^\theta m(t)\,dt-m(\theta)^2.
\]
Conditions \textnormal{(a)} and \textnormal{(b)} imply that \(m\) is
\(1\)-Lipschitz and hence absolutely continuous. Moreover,
\(\varphi(0)=0\) and
\[
\varphi'(\theta)
=2m'(\theta)[\theta-m(\theta)]\geq0
\qquad
\text{for almost every }\theta,
\]
because \(m'\geq0\) almost everywhere and \(m(\theta)\leq\theta\). Thus,
condition \textnormal{(c)} gives \(q(\theta)=\varphi(\theta)\geq0\).

In addition, given \(m(0)=0\) and concavity, the restriction
\(m(\theta)\leq\theta\) for every \(\theta\) is equivalent to
\(m'_+(0)\leq1\). Hence this limited-liability-type slope restriction need
only be checked at the bottom of the state space.
\end{remark}

\begin{corollary}[Incentive-compatible mechanisms]
\label{cor:imple-IC}
Let \(\pi\in\Pi^{IC}\). Then \(\pi\) is payoff-equivalent to a random-cap
mechanism if and only if \(m^\pi\) is concave and \(U^\pi(0)=0\).
Equivalently, relative to Lemma~\ref{lem:IC}, random-cap representability adds
exactly two restrictions: concavity of the expected action and the boundary
normalization \(U^\pi(0)=0\).
\end{corollary}

\begin{proof}
If \(\pi\) is payoff-equivalent to a random cap, Proposition~\ref{prop:imple}
implies that \(m^\pi\) is concave. At \(\theta=0\), every random cap induces
the action zero with probability one, so \(U^\pi(0)=0\).

Conversely, suppose that \(\pi\in\Pi^{IC}\), \(m^\pi\) is concave, and
\(U^\pi(0)=0\). Because
\[
U^\pi(0)=-m^\pi(0)^2-q^\pi(0),
\]
and \(q^\pi(0)\geq0\), we have
\begin{equation}\label{eq:IC-random-cap-origin}
m^\pi(0)=q^\pi(0)=0.
\end{equation}
Lemma~\ref{lem:IC} implies that \(m^\pi\) is nondecreasing and that
\eqref{eq:KM-envelope} holds. Combining the envelope identity with
\(U^\pi(0)=0\) and
\[
U^\pi(\theta)
=-[\theta-m^\pi(\theta)]^2-q^\pi(\theta)
\]
gives
\begin{equation}\label{eq:IC-random-cap-variance}
q^\pi(\theta)
=2\theta m^\pi(\theta)
-2\int_0^\theta m^\pi(t)\,dt
-m^\pi(\theta)^2.
\end{equation}

It remains only to verify \(m^\pi(\theta)\leq\theta\). Concavity and
\eqref{eq:IC-random-cap-origin} imply
\[
\int_0^\theta m^\pi(t)\,dt
\geq\frac{\theta m^\pi(\theta)}{2}.
\]
Using this inequality in \eqref{eq:IC-random-cap-variance} yields
\[
0\leq q^\pi(\theta)
\leq m^\pi(\theta)[\theta-m^\pi(\theta)].
\]
Monotonicity and \(m^\pi(0)=0\) imply \(m^\pi(\theta)\geq0\), and therefore
\(m^\pi(\theta)\leq\theta\). All the conditions of
Proposition~\ref{prop:imple} now hold, so \(\pi\) is payoff-equivalent to a
random-cap mechanism.
\end{proof}
\section{Proofs for the MPC extension}
\label{app:MPC-proofs}

Throughout this appendix, write
\begin{align*}
S_F(x)&:=1-F(x),
&C_F(x)&:=\int_x^1 S_F(t)\,dt,
&e_F(x)&:=\frac{C_F(x)}{S_F(x)}.
\end{align*}
Thus, \(C_F\) is the call function of \(F\), and \(e_F\) is its mean
residual life.

\begin{proof}[Proof of \cref{lem:MPC-feasibility}]
We divide the argument into two parts. The first extracts the two
reduced-form implications of the increasing-hazard-rate assumption that are
used in the construction. The second derives the result from those
properties alone.

\paragraph{Part 1: reduced-form implications of an increasing hazard rate}
Define
\[
Q_b(x):=e^{\frac{x}{b}}C_F(x).
\]
We first show that \(Q_b\) is strictly single-peaked. An increasing hazard
rate implies decreasing mean residual life. Indeed,
\[
e_F(x)
=\int_x^1
e^{-\int_x^t h_F(s)\,ds}dt
<\frac{1}{h_F(x)},
\]
and hence
\[
e_F'(x)=h_F(x)e_F(x)-1<0.
\]
Moreover,
\begin{equation}\label{eq:Qb-derivative}
Q_b'(x)
=\frac{e^{\frac{x}{b}}S_F(x)}{b}\,[e_F(x)-b].
\end{equation}
Because \(e_F(0)=\mu>b\) and \(e_F(x)\to0\) as \(x\uparrow1\), the
function \(Q_b\) first increases and then decreases.

The same argument yields the contact condition used below. Let
\[
\Gamma_b(x):=x+e_F(x)+b\log S_F(x).
\]
Using \(S_F'(x)=-h_F(x)S_F(x)\) and
\(e_F'(x)=h_F(x)e_F(x)-1\), we obtain
\begin{equation}\label{eq:Gamma-derivative}
\Gamma_b'(x)=h_F(x)[e_F(x)-b].
\end{equation}
Thus, \(\Gamma_b\) has the same increasing and decreasing regions as
\(Q_b\). Since
\[
\Gamma_b(0)=\mu,
\qquad
\lim_{x\uparrow1}\Gamma_b(x)=-\infty,
\]
there is a unique \(\tau_F\in(0,1)\) on the decreasing branch such that
\begin{equation}\label{eq:MPC-contact-condition}
\tau_F+e_F(\tau_F)+b\log S_F(\tau_F)=\mu,
\qquad
e_F(\tau_F)<b.
\end{equation}
The second inequality is the transversality, or marginal, condition. Hence
Assumption~\ref{ass:regu} implies the following two reduced-form properties:
\(Q_b\) is strictly single-peaked, and a cutoff \(\tau_F\) satisfying
\eqref{eq:MPC-contact-condition} exists uniquely.

\paragraph{Part 2: construction from single-peakedness and the marginal
condition}
The remainder uses only these two reduced-form properties. Define
\begin{equation}\label{eq:zbar-from-tau}
\bar z_F
:=\mu-b-b\log S_F(\tau_F).
\end{equation}
The equality in
\eqref{eq:MPC-contact-condition} also gives
\[
\bar z_F=\tau_F+e_F(\tau_F)-b.
\]
Equivalently,
\begin{equation}\label{eq:MPC-root-equivalent}
\tau_F-\bar z_F+e_F(\tau_F)=b.
\end{equation}
It follows that
\[
\mu-b<\bar z_F<\tau_F<1,
\]
and
\begin{equation}\label{eq:MPC-gap-bounds}
0<\tau_F-\bar z_F=b-e_F(\tau_F)<b.
\end{equation}
Moreover,
\[
S_F(\tau_F)
=e^{-\frac{\bar z_F-(\mu-b)}{b}}.
\]
Since \(S_F\) is strictly decreasing, this \(\tau_F\) coincides with the
cutoff defined in \eqref{eq:tau-def}.

We next verify the root condition. The preceding identities imply
\begin{align*}
H_F(\bar z_F)+b
&=\int_{\bar z_F}^{\tau_F}1\,dx
 +\frac{1}{S_F(\tau_F)}
  \int_{\tau_F}^1S_F(x)\,dx\\
&=\tau_F-\bar z_F+e_F(\tau_F)
=b.
\end{align*}
Thus, \(H_F(\bar z_F)=0\).

For completeness, this root is unique. For each \(z\in[\mu-b,1]\), let
\(y(z)\) be defined by
\[
S_F(y(z))
=e^{-\frac{z-(\mu-b)}{b}}
\]
and write \(J(z):=H_F(z)+b\). Differentiating separately according as
\(y(z)\leq z\) or \(y(z)>z\) gives
\[
\operatorname{sgn}J'(z)
=
\operatorname{sgn}
\left(e_F\bigl(\max\{z,y(z)\}\bigr)-b\right).
\]
By \eqref{eq:Qb-derivative} and the strict single-peakedness of \(Q_b\),
the expression on the right changes sign at most once, from positive to
negative. Hence \(J\) is strictly single-peaked. Finally,
\[
J(\mu-b)=C_F(\mu-b)
=b+\int_0^{\mu-b}F(t)\,dt>b,
\qquad
J(1)=0,
\]
so \(J(z)=b\), and therefore \(H_F(z)=0\), has exactly one solution.

The pieces in \eqref{eq:GF-def} now join to form a valid survival function.
The root condition also yields
\begin{equation}\label{eq:MPC-tail-mean}
\int_{\bar z_F}^1
\min\{S_F(\tau_F),S_F(x)\}\,dx
=bS_F(\tau_F).
\end{equation}
Consequently,
\begin{align*}
\mathbb E_{G^F}[\theta]
&=\mu-b
+\int_{\mu-b}^{\bar z_F}e^{-\frac{x-(\mu-b)}{b}}\,dx
+\int_{\bar z_F}^1
\min\{S_F(\tau_F),S_F(x)\}\,dx\\
&=\mu-b+b[1-S_F(\tau_F)]+bS_F(\tau_F)
=\mu.
\end{align*}

It remains to establish convex-order feasibility. For any distribution \(K\)
on \([0,1]\), let
\[
C_K(t):=\int_t^1[1-K(x)]\,dx,
\qquad
\Delta(t):=C_F(t)-C_{G^F}(t).
\]
For \(t\geq\tau_F\), the two call functions coincide. For
\(t\in[\bar z_F,\tau_F]\),
\[
\Delta(t)
=\int_t^{\tau_F}[S_F(x)-S_F(\tau_F)]\,dx
\geq0.
\]
For \(t\in[\mu-b,\bar z_F]\), \eqref{eq:MPC-tail-mean} gives
\[
C_{G^F}(t)
=be^{-\frac{t-(\mu-b)}{b}}.
\]
Therefore, \(C_{G^F}(t)\leq C_F(t)\) is equivalent to
\[
Q_b(t)\geq b e^{\frac{\mu-b}{b}}.
\]
This inequality holds at both endpoints:
\[
Q_b(\mu-b)
=e^{\frac{\mu-b}{b}}C_F(\mu-b)>b e^{\frac{\mu-b}{b}},
\]
and
\[
Q_b(\bar z_F)-b e^{\frac{\mu-b}{b}}
=e^{\frac{\bar z_F}{b}}
\int_{\bar z_F}^{\tau_F}
[S_F(x)-S_F(\tau_F)]\,dx>0.
\]
Because \(Q_b\) is single-peaked, its minimum on
\([\mu-b,\bar z_F]\) is attained at an endpoint. Hence the inequality holds
throughout this interval.

Finally, equal means imply \(\Delta(0)=0\), while for
\(t\in[0,\mu-b]\),
\[
\Delta'(t)=1-S_F(t)=F(t)\geq0.
\]
Thus, \(C_{G^F}(t)\leq C_F(t)\) for every \(t\in[0,1]\). The call-function
characterization of convex order yields
\(G^F\preceq_{\mathrm{cx}}F\).
\end{proof}

\begin{proof}[Proof of \cref{thm:MPC-saddle}]
Set
\[
\kappa_F
:=1-F(\tau_F)
=e^{-\frac{\bar z_F-(\mu-b)}{b}}.
\]
By \eqref{eq:MPC-cutoff-properties}, \(\eta_F\in(1/2,1)\). Hence
\(R^F\) is a valid cdf supported on \([\mu-b,\bar z_F]\), and the induced
mechanism \(\pi^F\) is incentive compatible.

We first show that \(G^F\) bounds the payoff of every
incentive-compatible mechanism from above. Define
\[
\nu_{G^F}(dx)
:=[1-G^F(x)]\,dx-b\,dG^F(x).
\]
Since \(\operatorname{supp}(G^F)\subseteq[\mu-b,1]\),
Lemma~\ref{lem:EV-formula}, applied with \(G=G^F\), gives
\begin{equation}\label{eq:MPC-payoff-decomposition}
\begin{aligned}
W(\pi,G^F)
={}&U^\pi(\mu-b)+(\mu-b)^2
+2\int_{[\mu-b,1]}m^\pi(x)\,\nu_{G^F}(dx)\\
&\quad-\mathbb E_{G^F}[(\theta-b)^2].
\end{aligned}
\end{equation}

The two components of \(\nu_{G^F}\) cancel on
\([\mu-b,\bar z_F]\). For \(x\in[\bar z_F,\tau_F]\),
\eqref{eq:MPC-root-equivalent} gives
\begin{align*}
\nu_{G^F}([x,1])
&=\kappa_F(\tau_F-x)+C_F(\tau_F)-b\kappa_F\\
&=\kappa_F(\bar z_F-x)
\leq0.
\end{align*}
For \(x\in[\tau_F,1]\),
\[
\nu_{G^F}([x,1])
=C_F(x)-b[1-F(x)]
=[1-F(x)][e_F(x)-b]
\leq0,
\]
where the inequality follows because \(Q_b\) is nonincreasing on
\([\tau_F,1]\); see \eqref{eq:Qb-derivative}. Thus,
\[
\nu_{G^F}([x,1])\leq0
\qquad
\text{for every }x\in[\mu-b,1],
\]
with equality at \(x=\mu-b\).

Incentive compatibility implies that \(m^\pi\) is nondecreasing.
Stieltjes integration by parts therefore yields
\[
\int_{[\mu-b,1]}m^\pi(x)\,\nu_{G^F}(dx)
=
\int_{(\mu-b,1]}
\nu_{G^F}([x,1])\,dm^\pi(x)
\leq0.
\]
Moreover, \(U^\pi(\mu-b)\leq0\). It follows from
\eqref{eq:MPC-payoff-decomposition} and
\(\mathbb E_{G^F}[\theta]=\mu=(\mu-b)+b\) that
\begin{align}
W(\pi,G^F)
&\leq
(\mu-b)^2-\mathbb E_{G^F}[(\theta-b)^2]
\nonumber\\
&=-\operatorname{Var}_{G^F}(\theta).
\label{eq:MPC-upper-bound}
\end{align}
Under \(\pi^F\), the action at \(\mu-b\) equals \(\mu-b\), so
\(U^{\pi^F}(\mu-b)=0\). Moreover, \(m^{\pi^F}\) is constant on
\([\bar z_F,1]\), because every cap lies weakly below \(\bar z_F\).
Hence the upper bound is attained:
\begin{equation}\label{eq:MPC-upper-value}
W(\pi^F,G^F)
=-\operatorname{Var}_{G^F}(\theta)
\geq W(\pi,G^F)
\qquad
\forall\,\pi\in\Pi^{IC}.
\end{equation}

We next show that \(G^F\) is a worst-case distribution against \(\pi^F\).
Define the affine function
\[
L_F(\theta)
:=
-b^2
+\kappa_F[b+e_F(\tau_F)]
(\theta-(\mu-b)).
\]
A direct calculation gives
\begin{equation}\label{eq:mpc-Vstar}
V^{\pi^F}(\theta)
=
\begin{cases}
-b^2,
&0\leq\theta\leq\mu-b,\\[4pt]
L_F(\theta),
&\mu-b\leq\theta\leq\bar z_F,\\[4pt]
L_F(\theta)
+(\theta-\bar z_F)(\tau_F-\theta),
&\bar z_F\leq\theta\leq1.
\end{cases}
\end{equation}
Indeed, the slope on \([\mu-b,\bar z_F]\) is
\[
2b\eta_F\kappa_F
=
\kappa_F[2b-(\tau_F-\bar z_F)]
=
\kappa_F[b+e_F(\tau_F)],
\]
where the last equality follows from
\eqref{eq:MPC-root-equivalent}.

Define
\begin{equation}\label{eq:mpc-psi}
\psi_F(\theta)
:=
\begin{cases}
L_F(\theta),
&0\leq\theta\leq\tau_F,\\[4pt]
V^{\pi^F}(\theta),
&\tau_F\leq\theta\leq1.
\end{cases}
\end{equation}
Equivalently,
\[
\psi_F(\theta)
=
L_F(\theta)
-(\theta-\bar z_F)(\theta-\tau_F)_+.
\]
Because
\[
(\theta-\bar z_F)(\theta-\tau_F)_+
\]
is convex, \(\psi_F\) is concave. Moreover,
\eqref{eq:mpc-Vstar} implies
\[
V^{\pi^F}(\theta)-\psi_F(\theta)
=
\begin{cases}
\kappa_F[b+e_F(\tau_F)](\mu-b-\theta),
&0\leq\theta<\mu-b,\\[4pt]
0,
&\mu-b\leq\theta\leq\bar z_F,\\[4pt]
(\theta-\bar z_F)(\tau_F-\theta),
&\bar z_F<\theta<\tau_F,\\[4pt]
0,
&\tau_F\leq\theta\leq1.
\end{cases}
\]
Thus, \(\psi_F\) is a concave minorant of \(V^{\pi^F}\).

For every \(G\preceq_{\mathrm{cx}}F\), concavity of \(\psi_F\) gives
\begin{equation}\label{eq:MPC-lower-bound}
\begin{aligned}
W(\pi^F,G)
&=\mathbb E_G[V^{\pi^F}(\theta)]\\
&\geq\mathbb E_G[\psi_F(\theta)]\\
&\geq\mathbb E_F[\psi_F(\theta)].
\end{aligned}
\end{equation}
It remains to identify the last expectation. Since
\[
\operatorname{supp}(G^F)
\subseteq
[\mu-b,\bar z_F]\cup[\tau_F,1],
\]
we have \(V^{\pi^F}=\psi_F\), \(G^F\)-almost surely. Furthermore, write
\[
\psi_F(\theta)=L_F(\theta)-q_F(\theta),
\qquad
q_F(\theta)
:=(\theta-\bar z_F)(\theta-\tau_F)_+.
\]
The distributions \(F\) and \(G^F\) have the same mean, so the affine
function \(L_F\) has the same expectation under both distributions.
Moreover, \(q_F\) vanishes below \(\tau_F\), while \(G^F\) coincides with
\(F\) on \([\tau_F,1]\). Therefore,
\[
\mathbb E_F[\psi_F(\theta)]
=
\mathbb E_{G^F}[\psi_F(\theta)]
=
W(\pi^F,G^F)
=
-\operatorname{Var}_{G^F}(\theta).
\]
Combining this equality with \eqref{eq:MPC-lower-bound} yields
\[
W(\pi^F,G)
\geq W(\pi^F,G^F)
\qquad
\forall\,G\preceq_{\mathrm{cx}}F.
\]
Together with \eqref{eq:MPC-upper-value}, this proves
\eqref{eq:MPC-saddle-ineq} and \eqref{eq:MPC-value}.
\end{proof}

\subsection{Proof of the signal-design theorem}
\label{app:proof-MPC-signal}

\begin{proof}[Proof of \cref{thm:MPC-signal}]
{
We first bound the unrestricted signal-design problem. Let
\(\widetilde\pi\in\widetilde\Pi^{IC}\) be arbitrary. Nature can choose the
fully uninformative signal \(\tau^{\mathrm{FP}}=\delta_F\). Its unique posterior
has mean \(\mu\) and variance \(\operatorname{Var}_F(\theta)\), so the
principal's payoff is
\[
-\operatorname{Var}_F(\theta)
-\mathbb E_{a\sim\widetilde\pi(\cdot\mid F)}[(\mu-b-a)^2]
\leq-\operatorname{Var}_F(\theta).
\]
This bounds the unrestricted robust value from above.

The report-independent rule \(a=\mu-b\) is incentive compatible in the
unrestricted posterior-report problem. Against any Bayes-plausible \(\tau\),
it gives the principal
\[
\mathbb E_\tau\!\left[-\sigma_s^2-(x_s-\mu)^2\right]
=-\operatorname{Var}_F(\theta),
\]
where the equality follows from the law of total variance and
\(\mathbb E_\tau[x_s]=\mu\). The unrestricted signal-design problem therefore
has value \(-\operatorname{Var}_F(\theta)\).

It remains to compute the value of \eqref{eq:signal-MPC}. For every
\(\pi\in\Pi^{IC}\), the feasible distribution \(\delta_\mu\preceq_{\mathrm{cx}}F\)
gives
\[
\mathcal J(\pi,\delta_\mu)
=-\operatorname{Var}_F(\theta)
-\mathbb E_{a\sim\pi(\cdot\mid\mu)}[(\mu-b-a)^2]
\leq-\operatorname{Var}_F(\theta).
\]
Conversely, the constant rule \(a(x)=\mu-b\) satisfies, for every
\(\bar G\preceq_{\mathrm{cx}}F\),
\[
\mathcal J(\pi,\bar G)
=-\operatorname{Var}_F(\theta)
+\operatorname{Var}_{\bar G}(x)
-\mathbb E_{\bar G}[(x-\mu)^2]
=-\operatorname{Var}_F(\theta).
\]
Thus, \eqref{eq:signal-maxmin} and \eqref{eq:signal-MPC} are value-equivalent.

Finally, let \(\pi^{\mathrm{cap}}\) denote the deterministic cap
\(z=\mu-b\), which induces \(a(x)=\min\{x,\mu-b\}\) at every posterior with
mean \(x\). This mechanism is feasible in the unrestricted class. Its action
loss satisfies
\[
(x-b-a(x))^2
=
\begin{cases}
b^2, & x\leq\mu-b,\\[3pt]
(x-\mu)^2, & x>\mu-b,
\end{cases}
\leq(x-\mu)^2.
\]
Thus, for every Bayes-plausible signal and its induced distribution \(\bar G\)
of posterior means,
\[
\mathcal J(\pi^{\mathrm{cap}},\bar G)
\geq
-\operatorname{Var}_F(\theta)
+\operatorname{Var}_{\bar G}(x)
-\mathbb E_{\bar G}[(x-\mu)^2]
=-\operatorname{Var}_F(\theta).
\]
The cap attains equality under \(\tau^{\mathrm{FP}}\). It is therefore robustly
optimal in both problems.
}
\end{proof}

\section{Proof of the deterministic benchmark}
\label{app:deterministic}

{
\paragraph{Action-space scope.}
\citet{HuLi2023} state their interval result for a compact action space spanned
by the agent's ideal actions, whereas we maintain $A=\mathbb R$. This difference
does not change the deterministic robust value. To see this, let $m$ be any
deterministic incentive-compatible rule, write
$V^m(\theta):=-[m(\theta)-(\theta-b)]^2$, and set $r:=m(1)$. Incentive
compatibility makes $m$ nondecreasing. If $r<0$, then $m\equiv r$, and cap zero
raises expected payoff by
\[
r^2-2r(\mu-b)>0.
\]
If $0\leq r<\mu$, type $r$ can obtain action $r$ by reporting one, so incentive
compatibility gives $m(r)=m(1)=r$. Nature may then choose
\[
G_r=\frac{1-\mu}{1-r}\delta_r
    +\frac{\mu-r}{1-r}\delta_1,
\]
which implies
\begin{equation}
\inf_{G\in\mathcal G_\mu}\mathbb E_G[V^m(\theta)]
\leq -b^2+(\mu-r)\bigl(r-(1-2b)\bigr).
\label{eq:action-range-bound}
\end{equation}
If $r\geq\mu$, an atomic mean-$\mu$ attack instead gives payoff at most
$-b^2$.\footnote{\rev{Any assigned action $a\in[0,\mu]$ is a fixed point,
$m(a)=a$, and can be paired with an upper fixed point (or with state one when
$r>1$). If there is no such action and $\delta_\mu$ does not already give the
bound, incentive compatibility implies $m(0)=m(\mu)=x<0$ and the first upward
jump is to some $y>\mu$ at $t=(x+y)/2\geq\mu$. A limiting mixture of state zero
and states just above $t$, with mean $\mu$, has loss above $b^2$ by
$x^2+\mu t+2b\mu+2x(b-\mu)-4b\mu x/t>0$.}}
Maximizing the right-hand side of \eqref{eq:action-range-bound}, together with
the other two cases, gives exactly \eqref{eq:optimal-det-cap-value}; the caps
identified below attain these bounds. Thus the unrestricted deterministic
problem and the cap benchmark are value-equivalent.
}

\paragraph{Problem formulation}
\rev{Combining the compact-benchmark interval reduction of \citet{HuLi2023}
with the scope check above, a robustly optimal deterministic rule can be chosen
to be an interval. It remains to normalize such intervals to the form $[0,z]$
with $z\in[0,1]$.}

To see the latter claim, let \(D=[\underline a,\bar a]\) be an interval. If
\(\underline a>0\), replacing \(\underline a\) by zero weakly improves the
principal's payoff state by state: for \(\theta<\underline a\), the agent's
action falls from \(\underline a\) to \(\theta\), which is closer to the
principal's ideal action \(\theta-b\). If \(\underline a\leq0\) and
\(\bar a\geq0\), the lower endpoint never binds. Upper endpoints above one are
irrelevant. Finally, if \(\bar a<0\), the interval induces the constant action
\(\bar a\), whereas the cap \(z=0\) raises the principal's expected payoff,
for every \(G\in\mathcal G_\mu\), by
\[
\bar a^2-2\bar a(\mu-b)>0.
\]
The deterministic benchmark therefore reduces to
\begin{equation}
V^D(\mu)
:=
\sup_{z\in[0,1]}V_z^D(\mu),
\qquad
V_z^D(\mu)
:=
\inf_{G\in\mathcal G_\mu}\phi_z(G).
\label{eq:deterministic-problem}
\end{equation}
The question is thus to determine nature's worst-case distribution against
each fixed cap and then optimize the resulting guarantee over \(z\).

\begin{theorem}[Deterministic benchmark]
\label{thm:deterministic-benchmark}
Under Assumption~\ref{ass:mu-range}, the robust payoff of a fixed cap
\(z\in[0,1]\) is
\begin{equation}
V_z^D(\mu)
=
\begin{cases}
-\Bigl((1-\mu)b^2+\mu(z+b-1)^2\Bigr),
& z\leq1-2b,\\[8pt]
-b^2,
& z>1-2b \ \text{and}\ \mu\leq z,\\[8pt]
-\left(
\dfrac{1-\mu}{1-z}b^2
+
\dfrac{\mu-z}{1-z}(z+b-1)^2
\right),
& 1-2b<z<\mu.
\end{cases}
\label{eq:fixed-cap-value}
\end{equation}
For each parameter region in \eqref{eq:fixed-cap-value}, a corresponding
worst-case distribution is
\[
G_z^D
=
\begin{cases}
(1-\mu)\delta_0+\mu\delta_1,
& z\leq1-2b,\\[8pt]
\delta_\mu,
& z>1-2b \ \text{and}\ \mu\leq z,\\[8pt]
\dfrac{1-\mu}{1-z}\delta_z
+
\dfrac{\mu-z}{1-z}\delta_1,
& 1-2b<z<\mu.
\end{cases}
\]

Consequently,
\begin{equation}
\argmax_{z\in[0,1]}V_z^D(\mu)
=
\begin{cases}
[1-2b,1],
& \mu\leq1-2b,\\[6pt]
\left\{\dfrac{\mu+1-2b}{2}\right\},
& \mu>1-2b,
\end{cases}
\label{eq:optimal-deterministic-cap}
\end{equation}
and
\begin{equation}
V^D(\mu)
=
\begin{cases}
-b^2,
& \mu\leq1-2b,\\[8pt]
-b^2+\dfrac{(\mu-(1-2b))^2}{4},
& \mu>1-2b.
\end{cases}
\label{eq:deterministic-value}
\end{equation}
Equivalently, in the second case,
\[
V^D(\mu)
=
-\left[
b(1-\mu)-\frac{(1-\mu)^2}{4}
\right].
\]
Moreover, the optimal random cap strictly dominates the deterministic
benchmark:
\[
V^D(\mu)<V^*.
\]
\end{theorem}

\begin{proof}
For a fixed cap \(z\), define the principal's loss
\[
u_z(\theta)
:=
\bigl(\min\{\theta,z\}-(\theta-b)\bigr)^2
=
\begin{cases}
b^2, & \theta\leq z,\\[4pt]
(z+b-\theta)^2, & \theta\geq z.
\end{cases}
\]
Since \(\phi_z(G)=-\mathbb E_G[u_z(\theta)]\), nature's problem is equivalent
to maximizing \(\mathbb E_G[u_z(\theta)]\) subject to
\(\mathbb E_G[\theta]=\mu\).

\medskip
\noindent
\emph{Step 1: The value of a fixed cap.}

Suppose first that \(z\leq1-2b\). Consider the affine function
\[
\ell_z(\theta)
:=
b^2+\theta\bigl((z+b-1)^2-b^2\bigr).
\]
For \(\theta\leq z\),
\[
\ell_z(\theta)-u_z(\theta)
=
\theta\bigl((z+b-1)^2-b^2\bigr)
\geq0.
\]
For \(\theta\geq z\),
\[
\ell_z(\theta)-u_z(\theta)
=
(1-\theta)\bigl(\theta-z(z+2b)\bigr)
\geq0,
\]
because \(z+2b\leq1\) implies \(z(z+2b)\leq z\leq\theta\). Thus,
\(u_z\leq\ell_z\) on \([0,1]\), and hence
\[
\mathbb E_G[u_z(\theta)]
\leq
\ell_z(\mu)
=
(1-\mu)b^2+\mu(z+b-1)^2.
\]
Equality is attained by \((1-\mu)\delta_0+\mu\delta_1\), since
\(\ell_z=u_z\) at both zero and one.

Now suppose that \(z>1-2b\). The case \(z=1\) is immediate, since
\(u_1(\theta)=b^2\) for every \(\theta\). For \(z<1\), define
\[
\overline u_z(\theta)
:=
\begin{cases}
b^2,
& \theta\leq z,\\[6pt]
b^2+
\dfrac{\theta-z}{1-z}
\bigl((z+b-1)^2-b^2\bigr),
& \theta\geq z.
\end{cases}
\]
The slope of \(\overline u_z\) to the right of \(z\) is
\[
\frac{(z+b-1)^2-b^2}{1-z}
=
1-z-2b<0.
\]
Thus, \(\overline u_z\) is concave. Moreover,
\[
\overline u_z(\theta)-u_z(\theta)
=
\begin{cases}
0, & \theta\leq z,\\[4pt]
(\theta-z)(1-\theta), & \theta\geq z,
\end{cases}
\]
so \(u_z\leq\overline u_z\) everywhere. Jensen's inequality therefore gives
\[
\mathbb E_G[u_z(\theta)]
\leq
\mathbb E_G[\overline u_z(\theta)]
\leq
\overline u_z(\mu).
\]

If \(\mu\leq z\), then \(\overline u_z(\mu)=b^2\), and equality is attained
by \(\delta_\mu\). If \(\mu>z\), then
\[
\overline u_z(\mu)
=
\frac{1-\mu}{1-z}b^2
+
\frac{\mu-z}{1-z}(z+b-1)^2,
\]
and equality is attained by
\[
\frac{1-\mu}{1-z}\delta_z
+
\frac{\mu-z}{1-z}\delta_1.
\]
This proves \eqref{eq:fixed-cap-value}.

\medskip
\noindent
\emph{Step 2: Optimization over caps.}

Let \(c:=1-2b\). For \(z\leq c\), equation
\eqref{eq:fixed-cap-value} implies
\[
V_z^D(\mu)\leq-b^2,
\]
with equality at \(z=c\). For \(z>c\) and \(z\geq\mu\), we also have
\(V_z^D(\mu)=-b^2\). Finally, for \(c<z<\mu\),
\begin{equation}
V_z^D(\mu)
=
-b^2+(\mu-z)(z-c).
\label{eq:fixed-cap-simplification}
\end{equation}

If \(\mu\leq c\), the interval \((c,\mu)\) is empty, so every
\(z\in[c,1]\) is optimal and the value is \(-b^2\). If \(\mu>c\), the
second term in \eqref{eq:fixed-cap-simplification} is a strictly concave
quadratic in \(z\), uniquely maximized at
\[
z_D^*
=
\frac{\mu+c}{2}
=
\frac{\mu+1-2b}{2}.
\]
Its maximum is \((\mu-c)^2/4\), proving
\eqref{eq:optimal-deterministic-cap} and
\eqref{eq:deterministic-value}.

It remains to compare the deterministic value with
\[
V^*
=
-b^2\left(1-e^{-(1-\mu)/b}\right).
\]
If \(\mu\leq1-2b\), then
\[
V^*-V^D(\mu)
=
b^2e^{-(1-\mu)/b}>0.
\]
If \(\mu>1-2b\), let \(x:=(1-\mu)/b\in(0,2)\). Then
\[
V^*-V^D(\mu)
=
b^2\left[
e^{-x}-\left(1-\frac{x}{2}\right)^2
\right]
>0,
\]
where the inequality follows from
\(e^{-x/2}>1-x/2>0\). Hence \(V^D(\mu)<V^*\) in both cases.
\end{proof}

\section{The general saddle point and the boundary regime}
\label{app:general-saddle}

\paragraph{Problem formulation}
Fix \(0<b<1\) and \(\mu\in[0,1]\). We consider the fixed-mean max--min and
min--max problems
\[
\sup_{\pi\in\Pi^{\mathrm{IC}}}
\inf_{G\in\mathcal G_\mu}W(\pi,G)
\qquad\text{and}\qquad
\inf_{G\in\mathcal G_\mu}
\sup_{\pi\in\Pi^{\mathrm{IC}}}W(\pi,G).
\]
The question is whether the mechanism \(\pi_\mu^*\) defined in
\eqref{eq:Rstar-all-mu}--\eqref{eq:boundary-A} and a suitable extension of the
shifted exponential-tail distribution form a saddle point for every mean,
including the boundary regime \(\mu<b\).

Define the cdf
\begin{equation}
H_\mu^*(\theta)
=
\begin{cases}
0,
&\theta<(\mu-b)_+,\\[4pt]
1-e^{-\frac{\theta-(\mu-b)_+}{b}},
&(\mu-b)_+\leq\theta<1-b,\\[8pt]
1-e^{-\frac{1-b-(\mu-b)_+}{b}},
&1-b\leq\theta<1,\\[8pt]
1,
&\theta\geq1,
\end{cases}
\label{eq:Hstar-all-mu}
\end{equation}
and let
\begin{equation}
G_\mu^*
=
\min\left\{\frac{\mu}{b},1\right\}H_\mu^*
+
\left(1-\min\left\{\frac{\mu}{b},1\right\}\right)\delta_0.
\label{eq:Gstar-all-mu}
\end{equation}
The distribution \(H_\mu^*\) has mean
\(b+(\mu-b)_+=\max\{\mu,b\}\), so \(G_\mu^*\) has mean \(\mu\). When
\(\mu\geq b\), the atom at zero vanishes and \(G_\mu^*=H_\mu^*\) is the
baseline shifted exponential-tail distribution. When \(\mu<b\),
\eqref{eq:Gstar-all-mu} becomes
\[
G_\mu^*
=
\frac{\mu}{b}G_b^*
+
\left(1-\frac{\mu}{b}\right)\delta_0,
\]
as stated in \eqref{eq:boundary-G-description}.

The candidate mechanism \(\pi_\mu^*\) draws a cap \(Z\), independently of
the state, from the cdf
\[
R_\mu^*(z)
=
\begin{cases}
0,
&z<(\mu-b)_+,\\[4pt]
\dfrac12e^{\frac{z-1+b}{b}},
&(\mu-b)_+\leq z<1-b,\\[8pt]
1,
&z\geq1-b.
\end{cases}
\]
After observing \(Z\), the agent chooses a base action \(y\in[0,Z]\), and
the final action is drawn from the lottery \(A_y^\mu\) in
\eqref{eq:boundary-A}.

\begin{proposition}[General saddle point]
\label{prop:general-saddle}
For every \(\mu\in[0,1]\), the distribution \(G_\mu^*\) belongs to
\(\mathcal G_\mu\), and the pair \((\pi_\mu^*,G_\mu^*)\) satisfies
\begin{equation}
W(\pi_\mu^*,G)
\geq
W(\pi_\mu^*,G_\mu^*)
\geq
W(\pi,G_\mu^*)
\label{eq:general-saddle}
\end{equation}
for every \(\pi\in\Pi^{\mathrm{IC}}\) and \(G\in\mathcal G_\mu\).
Consequently,
\[
\sup_{\pi\in\Pi^{\mathrm{IC}}}
\inf_{G\in\mathcal G_\mu}W(\pi,G)
=
\inf_{G\in\mathcal G_\mu}
\sup_{\pi\in\Pi^{\mathrm{IC}}}W(\pi,G)
=
V^*(\mu),
\]
where
\begin{equation}
V^*(\mu)
=
-\operatorname{Var}_{G_\mu^*}(\theta)
=
\begin{cases}
-\mu\left(
2b-\mu-b e^{-\frac{1-b}{b}}
\right),
&0\leq\mu<b,\\[8pt]
-b^2\left(
1-e^{-\frac{1-\mu}{b}}
\right),
&b\leq\mu\leq1.
\end{cases}
\label{eq:value-all-mu}
\end{equation}

Moreover,
\[
\argmin_{G\in\mathcal G_\mu}W(\pi_\mu^*,G)
=
\left\{
G\in\mathcal G_\mu:
\operatorname{supp}(G)
\subseteq
[(\mu-b)_+,1-b]\cup\{1\}
\right\}.
\]
Finally, if \(\mu<b\), every pure random-cap mechanism is strictly suboptimal.
More precisely,
\begin{equation}
\sup_{\pi\in\Pi_{\mathrm{RC}}}
\inf_{G\in\mathcal G_\mu}W(\pi,G)
\leq
V^*(\mu)-(b-\mu)^2
<
V^*(\mu).
\label{eq:pure-cap-strict}
\end{equation}
\end{proposition}

\begin{proof}
We proceed in four steps.

\medskip\noindent
\emph{Step 1: Feasibility of \(G_\mu^*\).}

The function in \eqref{eq:Hstar-all-mu} is a valid cdf. By the tail-integral
formula,
\[
\begin{aligned}
\mathbb E_{H_\mu^*}[\theta]
&=
(\mu-b)_+
+
\int_{(\mu-b)_+}^{1-b}
e^{-\frac{x-(\mu-b)_+}{b}}\,dx
+
b e^{-\frac{1-b-(\mu-b)_+}{b}}\\
&=
(\mu-b)_++b.
\end{aligned}
\]
It follows from \eqref{eq:Gstar-all-mu} that
\[
\mathbb E_{G_\mu^*}[\theta]
=
\min\left\{\frac{\mu}{b},1\right\}
\bigl[b+(\mu-b)_+\bigr]
=
\mu.
\]
Hence \(G_\mu^*\in\mathcal G_\mu\).

\medskip\noindent
\emph{Step 2: Incentive compatibility of \(\pi_\mu^*\).}

The outcome lottery in \eqref{eq:boundary-A} satisfies
\begin{equation}
\mathbb E[A_y^\mu]
=y-(\mu-b)_-,
\qquad
\mathbb E[(A_y^\mu)^2]
=y^2+(\mu-b)_-^2,
\qquad
\operatorname{Var}(A_y^\mu)
=2y(\mu-b)_-.
\label{eq:boundary-lottery-moments}
\end{equation}
These identities also hold under the convention \(A_y^\mu=\delta_y\) when
\((\mu-b)_-=0\).

If the true state is \(\theta\) and the selected base action is \(y\), then
\eqref{eq:boundary-lottery-moments} gives
\[
\begin{aligned}
\mathbb E\left[(\theta-A_y^\mu)^2\right]
&=
\theta^2-2\theta\bigl[y-(\mu-b)_-\bigr]
+y^2+(\mu-b)_-^2\\
&=
(\theta-y)^2
+2(\mu-b)_-\theta
+(\mu-b)_-^2.
\end{aligned}
\]
The last two terms do not depend on \(y\). Conditional on a realized cap
\(Z\), the agent therefore minimizes \((\theta-y)^2\) over \(y\in[0,Z]\), and
hence chooses
\[
y=\min\{\theta,Z\}.
\]
Thus, the additional outcome lottery is incentive neutral and the induced
direct mechanism \(\pi_\mu^*\) belongs to \(\Pi^{\mathrm{IC}}\).

\medskip\noindent
\emph{Step 3: \(G_\mu^*\) bounds every incentive-compatible mechanism from
above.}

Fix \(\pi\in\Pi^{\mathrm{IC}}\), and write \(m=m^\pi\) and \(q=q^\pi\). Since
\(\operatorname{supp}(G_\mu^*)\subseteq[(\mu-b)_+,1]\), apply the payoff
decomposition with \(x_0=(\mu-b)_+\). On
\([(\mu-b)_+,1-b)\), the continuous parts of the associated signed measure
cancel. The possible atom at the lower endpoint, the flat upper survival
function, and the atom at one give
\begin{align}
\nu_{G_\mu^*}(dx)
={}&
-(\mu-b)_-\delta_{(\mu-b)_+}(dx)\notag\\
&+
\min\left\{\frac{\mu}{b},1\right\}
e^{-\frac{1-b-(\mu-b)_+}{b}}
\mathbf 1_{[1-b,1)}(x)\,dx\notag\\
&-
b\min\left\{\frac{\mu}{b},1\right\}
e^{-\frac{1-b-(\mu-b)_+}{b}}
\delta_1(dx).
\label{eq:nu-all-mu}
\end{align}
The payoff decomposition therefore becomes
\begin{align}
W(\pi,G_\mu^*)
={}&
U^\pi((\mu-b)_+)
+(\mu-b)_+^2
-2(\mu-b)_-m((\mu-b)_+)\notag\\
&\quad+
2\min\left\{\frac{\mu}{b},1\right\}
e^{-\frac{1-b-(\mu-b)_+}{b}}
\left[
\int_{1-b}^1m(x)\,dx-bm(1)
\right]\notag\\
&\quad-
\mathbb E_{G_\mu^*}\bigl[(\theta-b)^2\bigr].
\label{eq:upper-general}
\end{align}

Incentive compatibility implies that \(m\) is nondecreasing. Therefore,
\[
\int_{1-b}^1m(x)\,dx-bm(1)
=
\int_{1-b}^1[m(x)-m(1)]\,dx
\leq0.
\]
Moreover, using
\((\mu-b)_+-(\mu-b)_-=\mu-b\),
\[
\begin{aligned}
&U^\pi((\mu-b)_+)
+(\mu-b)_+^2
-2(\mu-b)_-m((\mu-b)_+)\\
&\qquad=
(\mu-b)^2
-\bigl[m((\mu-b)_+)-(\mu-b)\bigr]^2
-q((\mu-b)_+)\\
&\qquad\leq
(\mu-b)^2.
\end{aligned}
\]
Substituting these inequalities into \eqref{eq:upper-general} yields
\[
W(\pi,G_\mu^*)
\leq
(\mu-b)^2
-
\mathbb E_{G_\mu^*}\bigl[(\theta-b)^2\bigr]
=
-\operatorname{Var}_{G_\mu^*}(\theta).
\]
Thus,
\begin{equation}
W(\pi,G_\mu^*)
\leq
-\operatorname{Var}_{G_\mu^*}(\theta)
\qquad
\text{for every }\pi\in\Pi^{\mathrm{IC}}.
\label{eq:upper-all-mu}
\end{equation}

Under \(\pi_\mu^*\), the base action at
\(\theta=(\mu-b)_+\) equals \((\mu-b)_+\). Since
\((\mu-b)_+(\mu-b)_-=0\), the final action is degenerate at
\[
(\mu-b)_+-(\mu-b)_-=\mu-b.
\]
Consequently,
\[
m^{\pi_\mu^*}((\mu-b)_+)=\mu-b,
\qquad
q^{\pi_\mu^*}((\mu-b)_+)=0.
\]
In addition, all caps are weakly below \(1-b\), so
\(m^{\pi_\mu^*}\) is constant on \([1-b,1]\). Both inequalities above are
therefore equalities under \(\pi_\mu^*\), and
\begin{equation}
W(\pi_\mu^*,G_\mu^*)
=
-\operatorname{Var}_{G_\mu^*}(\theta).
\label{eq:upper-attained-all-mu}
\end{equation}

\medskip\noindent
\emph{Step 4: \(G_\mu^*\) is a worst case against \(\pi_\mu^*\).}

For \(Z\sim R_\mu^*\), define
\[
M_\mu(\theta)
:=\mathbb E[\min\{\theta,Z\}],
\qquad
S_\mu(\theta)
:=\mathbb E[\min\{\theta,Z\}^2].
\]
Consider the affine function
\begin{align}
L_\mu(\theta)
:={}&
-\bigl[\mu-(\mu-b)_+\bigr]^2\notag\\
&+
\left[
b e^{-\frac{1-b-(\mu-b)_+}{b}}
-2(\mu-b)_-
\right]
\bigl[\theta-(\mu-b)_+\bigr].
\label{eq:affine-all-mu}
\end{align}

By \eqref{eq:boundary-lottery-moments}, the interim payoff under
\(\pi_\mu^*\) is
\begin{equation}
V^{\pi_\mu^*}(\theta)
=
-\bigl[\theta-\mu+(\mu-b)_+\bigr]^2
+2(\theta-b)M_\mu(\theta)
-S_\mu(\theta).
\label{eq:interim-general}
\end{equation}
For \((\mu-b)_+\leq\theta<1-b\),
\[
R_\mu^*(\theta)
=
\frac12e^{\frac{\theta-1+b}{b}},
\]
and hence
\[
\begin{aligned}
M_\mu(\theta)
&=
\theta-bR_\mu^*(\theta)
+\frac b2e^{-\frac{1-b-(\mu-b)_+}{b}},\\
M_\mu'(\theta)
&=1-R_\mu^*(\theta),\\
S_\mu'(\theta)
&=2\theta[1-R_\mu^*(\theta)].
\end{aligned}
\]
Differentiating \eqref{eq:interim-general} therefore gives
\[
\frac{d}{d\theta}V^{\pi_\mu^*}(\theta)
=
b e^{-\frac{1-b-(\mu-b)_+}{b}}
-2(\mu-b)_-.
\]
At \(\theta=(\mu-b)_+\), the final action is degenerate at \(\mu-b\), so
\[
V^{\pi_\mu^*}((\mu-b)_+)
=
-\bigl[\mu-(\mu-b)_+\bigr]^2.
\]
It follows that
\[
V^{\pi_\mu^*}(\theta)=L_\mu(\theta)
\qquad
\text{for }\theta\in[(\mu-b)_+,1-b].
\]

If \(\theta<(\mu-b)_+\), then necessarily \(\mu\geq b\), every cap exceeds
\(\theta\), and the induced action is \(a=\theta\). Hence
\(V^{\pi_\mu^*}(\theta)=-b^2\), which gives
\[
V^{\pi_\mu^*}(\theta)-L_\mu(\theta)
=
b e^{-\frac{1-b-(\mu-b)_+}{b}}
\bigl[(\mu-b)_+-\theta\bigr].
\]

Finally, for \(\theta>1-b\), the base action equals \(Z\), so \(M_\mu\) and
\(S_\mu\) are constant. Moreover,
\[
M_\mu(1-b)
=
1-b-\frac b2
+\frac b2e^{-\frac{1-b-(\mu-b)_+}{b}}.
\]
The difference \(V^{\pi_\mu^*}-L_\mu\) satisfies
\[
\bigl[V^{\pi_\mu^*}-L_\mu\bigr](1-b)=0,
\qquad
\bigl[V^{\pi_\mu^*}-L_\mu\bigr]'(1-b+)=b,
\qquad
\bigl[V^{\pi_\mu^*}-L_\mu\bigr]''(\theta)=-2.
\]
Therefore,
\[
V^{\pi_\mu^*}(\theta)-L_\mu(\theta)
=
(\theta-(1-b))(1-\theta)
\qquad
\text{for }\theta\in[1-b,1].
\]

Combining the three regions,
\begin{equation}
V^{\pi_\mu^*}(\theta)-L_\mu(\theta)
=
\begin{cases}
b e^{-\frac{1-b-(\mu-b)_+}{b}}
\bigl[(\mu-b)_+-\theta\bigr],
&0\leq\theta<(\mu-b)_+,\\[6pt]
0,
&(\mu-b)_+\leq\theta\leq1-b,\\[4pt]
(\theta-(1-b))(1-\theta),
&1-b<\theta\leq1.
\end{cases}
\label{eq:global-minorant-all-mu}
\end{equation}
Thus, \(L_\mu\) is a global affine minorant of \(V^{\pi_\mu^*}\), with
contact set
\[
[(\mu-b)_+,1-b]\cup\{1\}.
\]

For every \(G\in\mathcal G_\mu\),
\[
W(\pi_\mu^*,G)
=
\mathbb E_G[V^{\pi_\mu^*}(\theta)]
\geq
\mathbb E_G[L_\mu(\theta)]
=
L_\mu(\mu).
\]
Since \(G_\mu^*\) is supported on the contact set,
\[
W(\pi_\mu^*,G_\mu^*)=L_\mu(\mu).
\]
Together with \eqref{eq:upper-attained-all-mu}, this gives
\[
L_\mu(\mu)
=
-\operatorname{Var}_{G_\mu^*}(\theta).
\]
Equation \eqref{eq:global-minorant-all-mu} also shows that equality holds if
and only if \(G\) is concentrated on
\([(\mu-b)_+,1-b]\cup\{1\}\), proving the stated characterization of
nature's best responses.

It remains only to evaluate the value. If \(\mu\geq b\), then
\((\mu-b)_+=\mu-b\) and \((\mu-b)_-=0\), so
\[
L_\mu(\mu)
=
-b^2+b^2e^{-\frac{1-\mu}{b}}.
\]
If \(\mu<b\), then \((\mu-b)_+=0\) and \((\mu-b)_-=b-\mu\), so
\[
\begin{aligned}
L_\mu(\mu)
&=
-\mu^2
+
\mu\left[
b e^{-\frac{1-b}{b}}-2(b-\mu)
\right]\\
&=
-\mu\left(
2b-\mu-b e^{-\frac{1-b}{b}}
\right).
\end{aligned}
\]
This proves \eqref{eq:value-all-mu}. Combining the upper and lower bounds
proves the saddle inequalities \eqref{eq:general-saddle} and the
max--min/min--max equality.

\medskip\noindent
\emph{Strict suboptimality of pure random caps when \(\mu<b\).}

Let \(\pi\in\Pi_{\mathrm{RC}}\) be any pure random-cap mechanism. At the
bottom state it implements action zero with probability one, so
\[
m^\pi(0)=0,
\qquad
U^\pi(0)=0.
\]
When \(\mu<b\), we have \((\mu-b)_+=0\) and \((\mu-b)_-=b-\mu\).
Substituting these boundary conditions into \eqref{eq:upper-general}, and
again using monotonicity of \(m^\pi\), gives
\[
W(\pi,G_\mu^*)
\leq
-\mathbb E_{G_\mu^*}\bigl[(\theta-b)^2\bigr].
\]
Since \(\mathbb E_{G_\mu^*}[\theta]=\mu\),
\[
-\mathbb E_{G_\mu^*}\bigl[(\theta-b)^2\bigr]
=
-\operatorname{Var}_{G_\mu^*}(\theta)
-(b-\mu)^2
=
V^*(\mu)-(b-\mu)^2.
\]
Therefore,
\[
\inf_{G\in\mathcal G_\mu}W(\pi,G)
\leq
W(\pi,G_\mu^*)
\leq
V^*(\mu)-(b-\mu)^2
<
V^*(\mu).
\]
Taking the supremum over \(\pi\in\Pi_{\mathrm{RC}}\) proves
\eqref{eq:pure-cap-strict}.
\end{proof}

\clearpage
\bibliographystyle{plainnat}
\bibliography{references}

\end{document}